\documentclass[11pt]{article}

\usepackage{graphicx,epsfig,amsthm,amsmath,amscd,amssymb}
\usepackage{setspace}
\usepackage{latexsym}

\newcommand{\n}[1] {\mbox{\boldmath{$#1$}}} 
\newcommand{\be}{\begin{eqnarray}}
\newcommand{\ee}{\end{eqnarray}}
\newcommand{\beq}[1]{\begin{equation}\label{#1}}
\newcommand{\eeq}{\end{equation}}
\newcommand{\ba}{\begin{eqnarray*}}
\newcommand{\ea}{\end{eqnarray*}}

\newcommand{\uq}{\underline{q}}

\newtheorem{deft}{Definition}[section]
\newtheorem{lem}{Lemma}[section]
\newtheorem{prop}{Proposition}
\newtheorem{thm}{Theorem}

\title{\bf Generalization of Jeffreys' Divergence Based Priors\\ for Bayesian Hypothesis testing}
\author{
M.J. Bayarri\\ University of Valencia \and G. Garc\'{\i}a-Donato \\  University of Castilla-La
Mancha
\footnote{Address for correspondence: Gonzalo Garc\'{\i}a-Donato, Department of Economy, Plaza
Universidad 2, 02071 Albacete, Spain. Email:Gonzalo.GarciaDonato@uclm.es} \vspace{5mm} }

\makeindex

\begin{document}
\maketitle 


\begin{abstract}
In this paper we introduce objective proper prior distributions for hypothesis testing and model
selection based on measures of divergence between the competing models; we call them {\it
divergence based} (DB) priors. DB priors have simple forms and desirable properties, like
information (finite sample) consistency; often, they are similar to other existing proposals like
the intrinsic priors; moreover, in normal linear models scenarios, they exactly reproduce
Jeffreys-Zellner-Siow priors.  Most importantly, in challenging scenarios such as irregular models
and mixture models, the DB priors are well defined and very reasonable, while alternative
proposals are not. We  derive approximations to the DB priors as well as MCMC and asymptotic
expressions for the associated Bayes factors.


\end{abstract}

\section{Introduction}
For the data $\n y$, with density $f(\n y\mid\n\theta,\n\nu)$, we consider the hypothesis testing
problem:
 \beq{ht} H_1:\n\theta=\n\theta_0,\hspace{.25cm}\mbox{vs.}
\hspace{.25cm} H_2:\n\theta\ne\n\theta_0, \eeq
 where $\n\theta_0\in\Theta$ is a known value. This is equivalent to the model selection
 problem of choosing between models:
  \beq{comp}
   M_1:f_1(\n{y}\mid\n\nu_1)=f(\n{y}\mid\n\theta_0,\n\nu_1)\hspace{.25cm}\mbox{vs.} \hspace{.25cm}
   M_2:f_2(\n{y}\mid\n\theta,\n\nu_2)=f(\n{y}\mid\n\theta,\n\nu_2),
   \eeq
where the notation reflects the fact that often  $\n\nu_1$ and $\n\nu_2$ represent different
quantities in each model. In Jeffreys' scenarios (Jeffreys, 1961), $\n\nu_1$ and $\n\nu_2$ had the
same meaning;  he called  $\n\theta$ the  {\it new parameter}, and $\n\nu_1$ and $\n\nu_2$, the
{\it  common parameters}  (also known as nuisance parameters). We revisit this issue in Section~4.

We aim for an {\it objective Bayes} solution to this model selection problem; that is,  no `external'
(subjective) information is assumed, other than the data,  $\n y$, and the information implicitly
needed to pose the problem, choose the competing models, etc. An excellent exposition of the
advantages of Bayesian methods, specially objective Bayes methods,  for problems with model
uncertainty is Berger and Pericchi (2001).

Usual Bayesian solutions (for $0$-$k_i$ loss functions) to (\ref{ht}) (or, equivalently, to
(\ref{comp})) are based on the posterior odds:
\[
\frac{\mbox{Pr}(H_1 \mid \n y)}{\mbox{Pr}(H_2 \mid \n y)} =  \, \frac{\mbox{Pr}(H_1)}{\mbox{Pr}(H_2)}
\ \times B_{12} \, ,
\]
where $\mbox{Pr}(H_i), \ i=1,2$ are the prior probabilities of the hypotheses, and  $ B_{12}$ is {\it
Bayes Factor} for $H_1$ against $H_2$:
\beq{bf}
 B_{12} = \frac{m_1(\n y)}{m_2(\n y)} = \frac{\int\, f_1(\n y\mid\n\nu_1) \, \pi_1(\n\nu_1)\,
d\n\nu_1}{\int\, f_2(\n y\mid\n\theta,\n\nu_2) \, \pi_2(\n\theta,\n\nu_2)\, d\n\theta  \, d\n\nu_2} \
, \eeq
where  $\pi_1(\n\nu_1)$ is the prior under $H_1$ and $\pi_2(\n\theta,\n\nu)$ the prior under $H_2$.
That is,  $B_{12}$ is the ratio of the marginal (averaged) likelihoods of the models.

It is common practice in objective Bayes approaches to concentrate on derivations of the Bayes
factors, letting the ultimate choice (whether objective or subjective) of the prior model
probabilities (and the derivations of the posterior odds) to the user. Bayes factors were
extensively used by Jeffreys (1961) as a measure of evidence in favor of a model (see also Berger,
1985; Berger and Delampady, 1987, and  Berger and Sellke, 1987); Kass and Raftery (1995) is a good
reference for review and applications. Bayes factors are also crucial ingredients of model
averaging approaches (see Clyde, 1999; Hoeting et al, 1999). In the rest of the paper, we
concentrate on the derivation of objective priors to compute Bayes factors.

 A main issue for deriving objective Bayes factors is appropriate choice of $\pi_1(\n\nu_1)$
 and $\pi_2(\n\theta,\n\nu_2)$ for use in \eqref{bf}. It is well known that familiar improper objective
 priors (or non-informative priors) for estimation problems (under a fixed model) are
 usually seriously inadequate in the presence of model uncertainty, generally producing arbitrary answers.
 (Interesting exceptions are studied in Berger, Pericchi and Varshavsky, 1998.) Of course, when
 improper priors can not be used, use of arbitrarily vague (but proper) priors is not a cure, and
 generally it is even worse. Another bad solution often encountered in practice is use of an apparently
 `innocuous', harmless, but yet arbitrary, proper prior, since it can severely dominate the likelihood in
ways that are not anticipated (and can not be investigated for high dimensional problems).

 There are two basic approaches to compute Bayes factors when there is not enough information available
 for trustworthy subjective assessment of $\pi_1(\n\nu_1)$ and $\pi_2(\n\theta,\n\nu_2)$ . A very successful one
 is to directly derive the objective Bayes factors themselves, usually by `training' and calibrating in several
 ways the non-appropriate Bayes factors obtained from usual objective improper priors (see Berger and Pericchi, 2001
 for reviews and  references). However, all these objective Bayes factors should ultimately be checked to
 correspond (approximately) to a genuine Bayes factor derived from a sensible prior. The alternative approach
 is to look for `formal rules' for constructing  `objective' but proper priors that have nice properties and are
 appropriate for using in model selection; Bayes factors are then just computed from these objective proper priors.
 Whether these Bayes factors are appropriate can then be directly judged from the adequacy of the
 priors used.

 Choice of prior distributions in scenarios of model uncertainty is still largely an open question, and
 only partial answers are known. Several methods have been proposed for use in general scenarios, like
 the arithmetic
intrinsic (AI) priors (Berger and Pericchi, 1996; Moreno, Bertolino
and Racugno, 1998); the fractional intrinsic (FI) priors (De Santis
and Spezaferri, 1999; Berger and Mortera, 1999); the expected
posterior (EP) priors (P\'{e}rez and Berger, 2002); the unit
information priors (Kass and Wasserman, 1995) and  predictively
matched priors (Ibrahim and Laud, 1994; Laud and Ibrahim, 1995;
Berger, Pericchi and Varshavsky, 1998; Berger and Pericchi, 2001).
In the specific context of linear models, widely used prior with
nice properties are Jeffreys-Zellner-Siow (JZS) priors (Jeffreys,
1961; Zellner and Siow, 1980,1984; Bayarri and Garc\'{\i}a-Donato,
2007). An interesting generalization is the mixtures of
$g$-priors (Liang et al., 2007).

All these methods are insightful, provide many interesting and useful ideas, and indeed have shown to behave nicely in a number of testing and model selection problems. Nonetheless, except for the very specific scenario of linear models, nobody seems to have investigated the ramifications of Jeffreys (1961) pioneering proposal (see the end of Section \ref{DBpriors}). His was indeed the first general derivation of objective priors for hypothesis testing, and was intended as a generalization of his proposal for testing a normal mean. Given the success of the generalization of this Jeffrey's testing prior to linear models (Zellner and Siow, 1980,1984; Bayarri and Garc\'{\i}a-Donato, 2007), it is somewhat surprising that his general proposal has not been pursued. We think that it is historically important to pursue this investigation, and we do so in this paper.

%
%

Specifically, we generalize Jeffrey's  pioneering suggestion, and use divergence measures between the competing models to
derive the required (proper) priors. We call these priors {\it
divergence based} (DB) priors. The main motivation was to generalize
the useful JZS priors for use in scenarios other than the normal
linear model, while at the same time extending Jeffrey's general proposal. We will show that indeed the DB priors are the JZS
priors in linear model contexts; also, they are as easy to derive
(often easier) than other popular proposals (AI, FI or EP priors),
being quite similar to them in many instances; most interestingly,
they are well defined in certain scenarios where all of the other proposals
fail.

For clarity of exposition, we consider first the case when there are no nuisance parameters.
Development for the general case is delayed till Section~\ref{nuisance}, once the basic ideas have
been introduced, and the behavior of DB priors studied in this considerably simpler scenario.

\section{DB priors}\label{DBpriors}
Assume first the problem without nuisance parameters:
%
\beq{compWN} M_1:f_1(\n{y})=f(\n{y}\mid\n\theta_0)\hspace{.25cm}\mbox{vs.} \hspace{.25cm}
M_2:f_2(\n{y}\mid\n\theta)=f(\n{y}\mid\n\theta). \eeq
 That is, the simpler model ($M_1$) involve no unknown parameters; hence only the prior for $\n \theta$
 under $M_2$ is needed. We drop the subindex in the previous section and denote such prior simply by
 $\pi(\n\theta)$; clearly $\pi(\n\theta)$ has to be proper.

Our proposal for DB priors for $\n \theta$ will be in terms of divergence measures between the
competing models $f(\n y \mid \n\theta_0)$ and $f(\n y\mid\n\theta)$, based on Kullback-Leibler
directed divergences
\begin{equation} \label{eq:KL}
 KL[\n\theta_0 :\n\theta]=\int
 [\log f(\n y \mid \n\theta) - \log f(\n y \mid \n \theta_0)] \,  f(\n y \mid \n\theta) \,d\n y,
 \end{equation}
(assuming continuous $\n y$ for simplicity). $KL$ is a measure of the information in $\n y$ to
discriminate between $\n \theta$ and $\n \theta_0$; it is designed to measure how far apart the
two competing densities are in the sense of the likelihood  (Schervish, 1995).

We do not directly use $KL$ to define the DB prior because it is not symmetric with respect to its arguments, and hence it would likely result in nonsymmetric priors; however, symmetric measures of divergence can be
derived by taking sums (which was Jeffrey's choice) or minimums of $KL$ divergences. We define:
 \begin{equation} \label{Ds}
 D^S[\n\theta,\n\theta_0]=
KL[\n\theta:\n\theta_0]+ KL[\n\theta_0:\n\theta],
\end{equation}
and
\begin{equation} \label{Dm}
D^M[\n\theta,\n\theta_0]= 2\times\;\min \{ KL[\n\theta:\n\theta_0], KL[\n\theta_0:\n\theta]\}.
 \end{equation}
We multiply by $2$ the minimum  in  the definition of $D^M$ so that both measures are in the same scale;
indeed, in some symmetric models (like in the normal scenario) both measures of
divergence coincide. Generalizations of $KL$, $D^S$ and $D^M$ to include marginal parameters are
discussed in Section~\ref{nuisance}. Note that $D^M$ is well defined even when one of directed
$KL$ divergences is not, which is the case when the competing models have different support. Except for these
irregular scenarios, $D^S$ is well defined and it is considerably easier to derive than $D^M$.
Most of the derivations and properties to follow are common to both $D^S$ and $D^M$. To avoid
tedious repetitions, we then simply use $D$ to refer to anyone of them. We use the superindex $S$
or $M$ only when necessary.

It is well known that $D \ge 0$ with equality if and only if $\n\theta=\n\theta_0$, although it is
not a metric (the triangle inequality does not hold). Our proposal, is based on
{\it unitary  measures of divergence}, 
$\bar{D}$,  which we take to be $D$ divided by the {\it effective
sample size} $n^*$, $\bar{D}= D/n^*$. In simple univariate i.i.d.
data  the effective sample equals the number of scalar data points,
but it does not need to be so in general. Indeed, in complex
situations, it can be a difficult concept; although there have been
several attempts in the literature to formalize it (see e.g. Pauler,
1998; Pauler, Wakefield and Kass, 1999; Berger et al. 2007), no
general agreed definition seems to exist. In all of the examples of
this paper, it is quite clear what $n^*$  should be, so we rely for
now in simple, intuitive interpretations.

\subsection{Motivation: scalar location parameters}
Suppose $\n y$ is a random sample from a univariate location family:
$$
f(\n y\mid\theta)=\prod_{i=1}^n\,f(y_i\mid\theta)=\prod_{i=1}^n\,
g(y_i-\theta),\hspace{0.25cm}\theta\in{\cal R}.
$$
It has been argued (Berger and Delampady 1987; Berger and Sellke
1987) that in symmetric problems with $\Theta={\cal R}$, objective
testing priors $\pi(\theta)$ under $H_2: \theta \neq \theta_0$
should be unimodal and symmetric about $\theta_0$; these priors
prevent introducing excessive bias toward $H_2$.  Accordingly, we
look for a proper $\pi(\theta)$ which, when in this simple scenario,  has these desirable
characteristics and which  is easily
generalizable to other situations.

As before, let  $\bar{D}$ be a {\it unitary} symmetrized divergence.
We consider use of a function, $h$ of  $\bar{D}$ as a testing prior
under $H_2$; that is  $\pi(\theta)\propto
h(\bar{D}[\theta,\theta_0])$. Since $\pi$ has to be proper, $h(t)$
has to be a decreasing (no-increasing) function for $t > 0$.
A first possibility could be to take $h(t)=\exp\{-qt\}$ for some
$q>0$, but this results in priors with short tails. Short-tailed
priors are usually not adequate for model selection, since they tend
to exhibit undesirable (finite sample) inconsistent behavior (see
Liang et al 2007).

We explore instead use of the functions $h_q(t)=(1+t)^{-q}$, where $q>0$  controls  thickness of the
tails of $\pi(\theta)$. Let
$$
c(q)=\int\, h_q(\bar{D}[\theta,\theta_0])\, d\theta=\int\,
\left(1+\frac{D[\theta,\theta_0]}{n^*}\right)^{-q}\, d\theta,
$$
and define
$$
\uq=\inf\{q\ge 0:\hspace{.1cm} c(q)<\infty\},\hspace{.7cm} q_*=\uq+1/2.
$$
For finite $\uq$, our specific proposal for a DB prior in this
location problem is
 \begin{equation} \label{dbloc}
 \pi^D(\theta)=c(q_*)^{-1}\,
 \left(1+\frac{D[\theta,\theta_0]}{n^*}\right)^{-q_*}\propto h_{q_*}
 \big(\bar{D}[\theta,\theta_0]\big) \, .
\end{equation}
 Generalization to vector valued $\n \theta$ is trivial.

We use $q_*$ instead of the more natural $\uq$ because $\uq$ is not guaranteed to produce proper priors. Of course, if $\uq$ is finite, any $q = \uq+\delta$, with $\delta>0$
results in proper priors, and hence could have been used to define a
DB prior. Our specific proposal, $\delta=1/2$ was chosen to
reproduce the well known  Jeffreys-Zellner-Siow prior in the Normal
context;  in general this choice results in densities with heavy
tails. Moreover, we have found that in general $0<\delta<1$ is a good choice since it produces
priors without moments, which in normal scenarios is needed to avoid
undesirable behavior of conjugate $g$ priors (Liang et al, 2007).

The following lemma establishes the desired symmetry and unimodality
of the DB prior. The proof follows easily from properties of $D$ in
these location problems and is avoided.
\begin{lem} Assume $\uq<\infty$; then  $\pi^{D}(\theta)$ is unimodal and  symmetric around $\theta_0$.
\end{lem}

 Definition of DB priors for scale parameters is also direct. Indeed
assume that  $\theta$ is a scale parameter for a positive random variable $X$; then, $\xi=\log\theta$
is a location parameter for $Y=\log X$, with density $f^*(y\mid\xi)$. Applying the definition in
\eqref{dbloc}, the DB prior for  $\xi$ is:

 \beq{dbscal}
\pi^D(\xi)\propto h_{q_*}(\bar{D}^*[\xi,\xi_0]),
 \eeq
where $\xi_0 = \log(\theta_0)$ and $\bar{D}^*[\xi,\xi_0]$ is the unitary measure of divergence
between $f^*(\n y\mid\xi_0)$ and $f^*(\n y\mid\xi)$. Therefore, in the original parameterization:
\beq{dbscal2}
 \pi^D(\theta)\propto h_{q_*}(\bar{D}^*[\log\theta,\log\theta_0]) \, \frac{1}{\theta}=
h_{q_*}(\bar{D}[\theta,\theta_0]) \, \pi^N(\theta),
 \eeq
 where, because of invariance of $\bar{D}$
under reparameterizations, $\bar{D}^*[\log\theta,\log\theta_0]=\bar{D}[\theta,\theta_0]$, and
$\pi^N(\theta)=1/\theta$ is the non informative prior (right Haar invariant prior) for $\theta$.
Definition of DB priors for general parameters, formalized in next section, will basically be a
generalization of \eqref{dbscal2}.

\subsection{General parameters}\label{GenPar}
Assume the more general problem (\ref{compWN}) and let
$\pi^N(\n\theta)$ be an objective (usually improper) `estimation'
prior (reference, invariant, Jeffreys, Uniform, ... prior) for
$\n\theta$, and let $\n\xi$ be a transformation such that
$\pi^N(\n\xi)=1$ for $\n\xi=\n\xi(\n\theta)$. We can then derive a
DB prior for  $\n \theta$ by considering $\n\xi$ as a ``location
parameter'', applying the definition \eqref{dbloc}, and transforming
back to $\n\theta$. This transformation was first proposed by
Jeffreys (1961). 
Bernardo (2005) uses it with a reference prior $\pi^N$  for a scalar
$\theta$, and notes that $\xi$ asymptotically behaves as a location
parameter.

 Giving $\n\xi$ a DB prior for location parameters results in:
\begin{equation}\label{lp}
\pi^D(\n\xi)\propto h_{q_*}(\bar{D}^*[\n\xi,\n\xi_0]),
\end{equation}
\noindent where, as before, $\bar{D}^*[\n\xi,\n\xi_0]$ denotes `unit' (symmetrized) discrepancy
between
 $f^*(\n y \mid \n\xi)$ and $f^*(\n y \mid \n\xi_0)$, and $\n\xi_0 = \n\xi(\n\theta_0)$. Hence,
 the corresponding ($DB$) prior for $\n\theta$ is
\begin{equation}\label{mot}
\pi^D(\n\theta)\propto h_{q_*}(\bar{D}^*[\n\xi(\n\theta),\n\xi(\n\theta_0)])\ |{\cal
J}_\theta(\n\theta)|\propto h_{q_*}(\bar{D}[\n\theta,\n\theta_0])\ \pi^N(\n\theta),
\end{equation}
as long as $\pi^N$ is invariant under transformations; ${\cal J}(\n\theta)$ is the jacobian of the
transformation.  It should be noted from \eqref{mot} that the explicit transformation to $\n\xi$
is not needed in order to derive the prior $\pi^D$. We can now formally define a DB prior as
follows:

\begin{deft}\label{GenPiD} {\rm \bf (General DB priors)} For the model selection problem \eqref{compWN},
let $\bar D[\n\theta,\n\theta_0]$ be a unitary measure of
divergence between $f(\n y\mid\n\theta)$ and $f(\n y\mid\n\theta_0)$. Also let $\pi^N(\n\theta)$
be an objective (possibly improper) estimation prior for $\n\theta$ under the complex model, $M_2$, and $h_q
(\cdot)$ be a decreasing function. Define:

$$
\uq=\inf\{q\ge 0:\hspace{.25cm} c(q)<\infty\},\hspace{.5cm} q_*=\uq+1/2,
$$
where $c(q)=\int\, h_q(\bar{D}[\n\theta,\n\theta_0])\,\pi^N(\n\theta)d\n\theta$. If $q_*<\infty$,
then a  divergence based prior under $M_2$ is defined as
\begin{equation} \label{eq:genDB}
\pi^{D}(\n\theta)=c(q_*)^{-1}\,h_{q_*}(\bar{D}[\n\theta,\n\theta_0])\ \pi^N(\n\theta).
\end{equation}
\end{deft}

Note that, by definition, the DB priors either do not exist, or they are proper (and hence they do not
involve arbitrary constants).

\vspace{2mm}

\paragraph{Specific Proposals.} Definition \ref{GenPiD} is very general, in that several definitions of
$\bar D$, $h_q$ and $\pi^N$ could be explored (as well as different choices of $0<\delta<1$ in  $q_*
= \uq +\delta$). We give specific choices which,  in part, are based on previous explorations and
desired properties of the resulting $\pi^D$; however our specific choices are mainly intended to
reproduce JZS  priors in normal scenarios, so that our proposals for DB priors can be best
contemplated as extensions of JZS priors to non-normal scenarios.

In what follows, we take $D$ to be either $D^S$ in \eqref{Ds} or $D^M$ in \eqref{Dm}, and
$h_q(t)=(1+t)^{-q}$. Since we will explore both, we need different notations:

\begin{deft}\label{piDsDm} {\rm \bf (Sum and Minimum DB priors)} The sum DB prior $\pi^S$ and the minimum DB
prior $\pi^M$ are the DB priors given in definition \ref{GenPiD} with $h_q(t)=(1+t)^{-q}$ and $D$
being respectively $D^S$ (see \eqref{Ds}) and $D^M$ (see \eqref{Dm}). When needed, we refer to their
corresponding c's and q's as $c_S, \uq^S, q_*^S$, and $c_M, \uq^M, q_*^M$, respectively.
\end{deft}



It can easily be shown that $c_S(q)\le c_M(q)$, so that, for regular problems (in which
$\bar{D}^S<\infty$), $q_*^M<\infty$ implies $q_*^s<\infty$, and therefore, in these problems,
existence of $\pi^{M}$ implies existence of $\pi^{S}$.

It should be noted that, although we are not explicitly assuming a specific objective prior $\pi^N$
in the definition of DB priors, properties of $\pi^N$  are inherited by the DB prior $\pi^D$; some
properties will be crucial for sensible DB priors, and hence appropriate choice of $\pi^N$ becomes
very important.

We now explore some appealing properties of DB priors. Since these
are common to both proposals in Definition \ref{piDsDm}, we drop
unneeded super and sub indexes and refer to the prior simply as
$\pi^D$. This convention will be kept through the paper; distinction
between $\pi^S$ and $\pi^M$ will only be done when needed.


\paragraph {Local behavior of DB priors.}  It can be easily checked that, when $\pi^N(\n\theta)=1$
(as when $\n\theta$ is a location parameter),  then the mode of $\pi^{D}$ is $\n\theta_0$ (so
$\pi^D$ is `centered' at the simplest model). We can also exploit the following (well known) approximate relationship between Kullback-Leibler divergence and Fisher information (see Kullback, 1968): for $\n\theta$ is in a neighborhood of $\n\theta_0$
$$
KL[\n\theta_0,\n\theta]\approx \frac{1}{2}(\n\theta-\n\theta_0)^t\n
J(\n\theta_0)(\n\theta-\n\theta_0),
$$
where $\n J(\n\theta_0)$ is the  expected
Fisher information matrix evaluated at $\n\theta_0$. Hence, in a neighborhood of $\n\theta_0$, 
the DB priors approximately behave as  $k$ multivariate Student distributions,
centered at $\n\theta_0$, and scaled by Fisher information matrix under the simpler model. That
is,
$$
\pi^{D}(\n\theta)\approx \mbox{St}_k(\n\theta_0,n^*\, \n J(\n\theta_0)^{-1}/d, \, d),
$$
\noindent where $d=2\uq-k+1$. Moreover, by definition of $q_*$, $d$ above is generally close to 1,
and then the DB priors would approximately be Cauchy.

As highlighted in Section~\ref{subsubLM}, the approximation above exactly holds in Normal
scenarios with $d = 1$, and hence the DB priors reproduce precisely the proposals of
Jeffreys-Zellner-Siow.  

%
%

\paragraph{Invariance under one-to-one transformations}
An important question is whether the DB priors are invariant under reparameterizations of the
problem. Suppose that $\n\xi=\n\xi(\n\theta)$ is a one-to-one monotone mapping
$\n\xi:\Theta\rightarrow\Theta_\xi$. The model selection problem (\ref{compWN}) now becomes:
\beq{compWNr} M_1^*:f_1^*(\n{y})=f^*(\n{y}\mid\n\xi_0)\hspace{.25cm}\mbox{vs.} \hspace{.25cm}
M_2^*:f_2^*(\n{y}\mid\n\xi)=f^*(\n{y}\mid\n\xi), \eeq where $f^*(\n y\mid\n\xi(\n\theta))=f(\n
y\mid\n\theta)$ and $\n\xi_0=\n\xi(\n\theta_0)$. The next result shows that, if $\pi^N$ is invariant
under the reparameterization $\n\xi(\n\theta)$ then so are the DB priors.

\begin{prop}\label{inv}
Let $\pi_\theta^{D}(\n\theta)$ and $\pi_\xi^{D}(\n\xi)$ denote the DB priors for the original
(\ref{compWN}), and reparameterized (\ref{compWNr}) problems respectively. If
$\pi_\theta^N(\n\theta)\propto \pi_\xi^N(\n\xi(\n\theta))|{\cal J}_\xi(\n\theta)|$, where ${\cal
J}_\xi$ is the Jacobian of the transformation
then
$$
\pi_\theta^{D}(\n\theta)= \pi_\xi^{D}(\n\xi(\n\theta))|{\cal J}_\xi(\n\theta)|.
$$
\end{prop}
\begin{proof} See Appendix.
\end{proof}

Under the conditions of  Proposition \ref{inv}, Bayes factors computed from DB priors are not
affected by reparameterizations. It is important to note that invariance of DB priors is a direct
consequence of both the invariance of the divergence measure used and the invariance of $\pi^N$. Some
objective priors $\pi^N$ invariant under reparameterizations are Jeffreys' priors and (partially) the
reference priors.

\paragraph{Compatibility with sufficient statistics.} DB priors are sometimes compatible with reduction of the data
via sufficient statistics. This attractive  property is not shared
by other objective Bayesian methods, as intrinsic Bayes factors.

\begin{prop}\label{CSS}
Let $\n t=\n t(\n y)$ be a sufficient statistic for $\n\theta$ in $f(\n y\mid\n\theta)$ with
distribution $f^*(\n t\mid\n\theta)$. Assume that $\pi^N$ and $n^*$ remain the same in the problem
defined by $f^*$, then the DB prior $\pi^{D}$ for the original problem (\ref{compWN}) is the same as the DB prior for the reduced (by sufficienty) testing problem
\begin{equation}\label{suf}
M_1^*:f_1^*(\n{t})=f^*(\n{t}\mid\n\theta_0)\hspace{.7cm}{vs.}\hspace{.7cm}
M_2^*:f_2^*(\n{t}\mid\n\theta)=f^*(\n{t}\mid\n\theta).
\end{equation}
\end{prop}
\begin{proof} See Appendix.
\end{proof}

\paragraph{DB priors and Jeffreys' general rule.}
Jeffreys (1961) proposed objective proper priors for testing
situations other than the normal mean. Specifically, when $\n y$ is
a random sample of size $n$, and for univariate $\theta$ he proposed
the following model testing prior:
\beq{Jefprior} \pi^J(\theta)=\frac{1}{\pi}\frac{d}{d\theta}\tan^{-1} \left( \frac{D^S[\theta,\theta_0]}{n}\right)^{1/2} =
\frac{1}{\pi} \left(1+ \frac{{D}^S[\theta,\theta_0]}{n}\right)^{-1}\,
\frac{d}{d\theta}\left(\frac{D^S[\theta,\theta_0]}{n}\right)^{1/2}
. \eeq

This reduces to Jeffreys Cauchy proposal when $\theta$ is a normal
mean. Also, when $|\theta-\theta_0|$ is small, $\pi^J(\theta)$ can
be approximated by
\beq{apJefprior} \pi^J(\theta) \approx
\frac{1}{\pi}\big(1+\bar{D}^S[\theta,\theta_0]\big)^{-1}\,
\pi^{NJ}(\theta), \eeq
where $\pi^{NJ}(\theta)$ is Jeffreys' (estimation) prior (i.e. the
squared root of the expected Fisher information).

Note that $\pi^J$ can lead to improper priors and at least in
principle can not be applied for multivariate parameters. However,
the approximation (\ref{apJefprior}) was a main inspiration for the
definition of DB priors, with clear similarities between them.

\section{Comparative examples: simple null}\label{examples}
In the spirit of Berger and Pericchi (2001) we investigate in this
section the performance of DB priors in a series of situations
chosen to be somehow representative of wider classes of statistical
problems. We also explicitly derive well established, alternative
proposals for objective priors in Bayesian hypothesis testing and
compare their performance with that of DB priors. We show that in
simple standard situations, DB priors produce similar results to
these alternative proposals. More interestingly, in more
sophisticated situations where these proposals fail (models with
irregular asymptotics or improper likelihoods), the DB priors are
well defined and very sensible.


We will compute and compare Bayes factors derived with DB priors,
with those derived with two of the most popular general objective priors for
objective Bayes model selection, namely:

\begin{enumerate}
\item Arithmetic intrinsic prior:
$$
\pi^A(\n\theta)=\pi^N(\n\theta)\, E^{M_2}_\theta\, (B_{12}^N(\n y^*)),
$$
where the Bayes factor $B^N$ is computed with the objective estimation prior $\pi^N$, and $\n y^*$ is
an imaginary sample of minimum size such that $0<m_2^N(\n y^*)<\infty$.
\item Fractional intrinsic prior:
$$
\pi^F(\n\theta)=\pi^N(\n\theta)\, \frac{\exp\{m \, E^{M_2}_\theta\log f(y\mid\n\theta_0)\}} {\int\,
\exp\{m \, E^{M_2}_\theta\log f(y\mid\tilde{\n\theta})\} \pi^N(\tilde{\n\theta})d\tilde{\n\theta}}.
$$
\end{enumerate}

In the iid case and asymptotically, $\pi^A$ produces the {\it arithmetic intrinsic Bayes factor}
(Berger and Pericchi, 1996), and $\pi^F$  the {\it fractional Bayes factor} (O'Hagan, 1995) if the
exponent of the likelihood is $b=m/n$ for a fixed $m$ (see De Santis and Spezaferri, 1999).
Following the recommendation of Berger and Pericchi (2001) we take $m$ to be the size of the
minimal training sample $\n y^*$.

In the examples of this Section, $\n y$ is an iid sample of size $n$ from $f(y\mid\theta)$, and
unless otherwise specified, $n^* = n$ ($n^*$ denotes effective sample size). We let $B_{12}^S$ denote the Bayes factor in favor of $H_1$
computed with $\pi^{S}$ (see Definition \ref{piDsDm}); $B_{12}^M, B_{12}^A$ and $B_{12}^F$ are
defined similarly.

\subsection{Bounded parameter space (Example 1)}

We begin with a simple example, in which data is a random sample from a Bernoulli distribution, that
is
$$
f(y\mid\theta)=\theta^{y}(1-\theta)^{1-y}, \hspace{.7cm}y\in\{0,1\},\,\,\,\,\, \theta\in\Theta=[0,1],
$$
and we want to test $M_1: \theta = \theta_0$ versus  $M_1: \theta \neq \theta_0$. The usual estimation objective prior (both reference and Jeffreys)
in this problem is the beta density $\pi^N(\theta)=Be(\theta\mid
1/2,1/2) \propto \theta^{-1/2}(1-\theta)^{-1/2}.$ In this case, since $\pi^N$ is proper, it would be tempting to use it as a testing prior. However, we will see that all $\pi^S, \pi^M, \pi^A$ and $\pi^F$ center around the null value $\theta_0$ whereas the estimation prior completely ignores it.

The DB prior for the sum-symmetrized divergence can be computed to be
$$
\pi^{S}(\theta)\propto
\Big[1+(\theta-\theta_0)\log\frac{\theta(1-\theta_0)}{\theta_0(1-\theta)}\Big]^{-1/2}\,
\pi^N(\theta),
$$
and the DB prior for the min-symmetrized divergence
$$
\pi^{M}(\theta)\propto \Big(1+\bar{D}^M[\theta,\theta_0]\Big)^{-1/2}\,\pi^N(\theta),
$$
where \ba \bar{D}^M[\theta,\theta_0]&=&\Big\{
\begin{array}{ccc}
2\,KL[\theta:\theta_0] & \mbox{if} &
\min\{\theta_0,1-\theta_0\}<\theta<\max\{\theta_0,1-\theta_0\}\\
2\,KL[\theta_0:\theta] & & \mbox{otherwise},
\end{array}
\ea and
$KL[\theta:\theta_0]=\theta_0\log\frac{\theta_0}{\theta}+(1-\theta_0)\log\frac{1-\theta_0}{1-\theta}$.

The intrinsic priors are derived in the next result. The proof is
straightforward and hence it is omitted.
\begin{lem} The arithmetic intrinsic prior is
$$
\pi^A(\theta)=\Big( \frac{2}{\pi} \ (1-\theta_0)(1-\theta)+\theta_0\theta \Big) \, \pi^N(\theta)
$$
and the fractional intrinsic prior is
$$
\pi^F(\theta)= 
\Big( \frac{\theta_0^\theta(1-\theta_0)^{1-\theta}} {\Gamma(\theta+1/2)\Gamma(3/2-\theta)}\Big) \,
\pi^N(\theta) .
$$
\end{lem}

By construction, $\pi^S$ and $\pi^M$ are proper priors; $\pi^A$ is
proper but $\pi^F$ is not. For instance, for $\theta_0=1/2$, $\pi^F$
integrates to 1.28 and for $\theta_0=3/4$, $\pi^F$  integrates to
1.18. This implies a small bias in the Bayes factor in favor of
$M_2$. In Figure~\ref{DBBern} we display $\pi^{S}, \pi^{M}, \pi^A$
and $\pi^F$ for $\theta_0=1/2$ and $\theta_0=3/4$. They can be seen
to be very similar. When $\theta_0=1/2$ they are also similar to the
objective estimation prior $Be(\theta \mid 1/2, 1/2)$, but not for
other values of $\theta_0$.

We also compute the Bayes factors for the four different priors,
when $\theta_0=1/2$, for two different sample sizes, $n=10$ and
$n=100$, and for different values of the MLE,
$\hat\theta=\sum_{i=1}^{10}\, y_i/n$ (see Table~\ref{TabBern}). All
the results are quite similar. As expected, $B_{12}^F$ gives the
most support to $M_2$; $B_{12}^A$ gives the least. Both DB priors
produce similar results, being slightly closer to $B_{12}^A$ than to
$B_{12}^F$.

Finally, we consider application to real data taken from Conover
(1971). Under the hypothesis of simple Mendelian inheritance, a
cross between two particular plants produces, in a proportion of
$\theta=3/4$ a specie called `giant'. To determine whether this
assumption is true, Conover (1971) crossed $n=925$ pair of plants,
getting $T=682$ giant plants. The Bayes factors in favor of the
Mendelian inheritance hypothesis (simplest model) are also given in
Table~\ref{TabBern} for the four different priors. Again the results
are very similar, the fractional intrinsic prior providing the least
support to $M_1$.

\begin{figure}[!t]
\begin{center}
\begin{tabular}{cc}
\includegraphics[width=175pt,height=120pt]{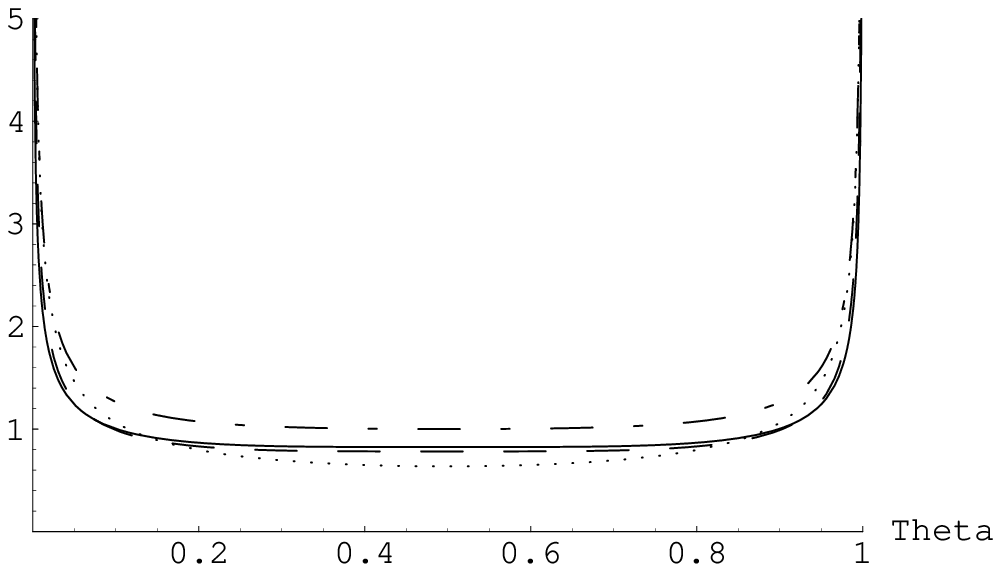} &
\includegraphics[width=175pt,height=120pt]{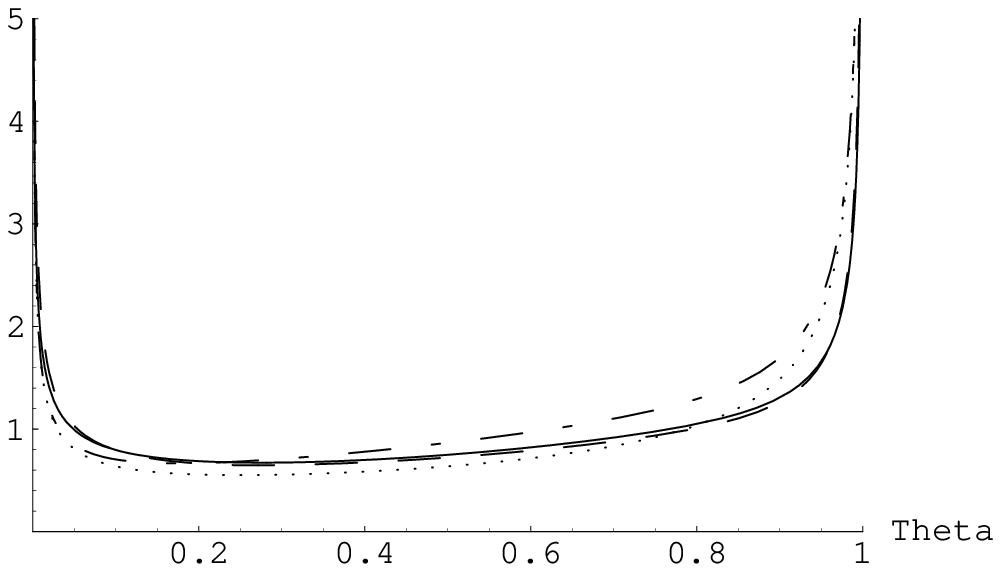}
\end{tabular}
\caption{In Bernoulli example: $\pi^S$ (Solid line), $\pi^M$
(Dot-dashed line), $\pi^A$ (Dots) and $\pi^F$ (Dashed line), for the
case $\theta_0=1/2$ (left) and $\theta_0=3/4$ (right).}
\label{DBBern}
\end{center}
\end{figure}

\begin{table}[t!]
\begin{center}{\small
\begin{tabular}{rc|ccccccc}
$n=10$ & $\hat\theta$ & $B_{12}^S$ & $B_{12}^M$ &  $B_{12}^A$ & $B_{12}^F$\\
\hline
&0.50 & 3.26  & 3.44 & 4.06 & 2.68\\
&0.65 & 2.14  & 2.24 & 2.58 & 1.75\\
&0.80 & 0.55  & 0.57 & 0.60 & 0.44\\
$n=100$ & &&&&&&&\\
&0.50 & 9.74 & 10.28 & 12.56 & 8.03 \\
&0.55 & 5.93 & 6.26 &   7.61 & 4.89  \\
&0.60 & 1.33 & 1.40 &   1.68 & 1.09 \\
Conover & &&&&&&&\\
&     & 19.38 &  20.20 & 20.79 & 16.02 \\
\hline
\end{tabular}\caption{Bayes factors in favor of $M_1$ for Bernoulli testing of $\theta_0=1/2$,
for different values of the MLE and $n=10$, $n=100$. Also, Bayes
factors for Conover data.} \label{TabBern}}
\end{center}
\end{table}

\subsection{Scale parameter (Example 2)}
We next consider another simple example of testing a scale parameter. Specifically, we consider that
data come from the one parameter exponential model with mean $\mu$, that is,
$$
f(y\mid \mu)= Exp \, (y \mid \frac{1}{\mu}) = \frac{1}{\mu}\, \exp\{- \, \frac{y}{\mu}\},\hspace{1cm}
y>0,  \ \ \mu>0,
$$
and that it is desired to test $H_1:\mu=\mu_0 \ $ vs. $ \ H_2:\mu\ne\mu_0$. Here
$\pi^N(\mu)=\mu^{-1}$, and the DB priors are computed to be:
 \ba \pi^{S}(\mu)\propto  \big[1+\frac{(\mu-\mu_0)^2}{\mu\mu_0}\big]^{-1/2}\mu^{-1}, \hspace{.5cm}
\pi^{M}(\mu)\propto \big(1+\bar{D}^M[\mu,\mu_0]\big)^{-3/2}\mu^{-1},
 \ea
 where
 \ba
 \bar{D}^M[\mu,\mu_0]&=&\Big\{
\begin{array}{ccc}
2\,KL[\mu_0:\mu] & \mbox{if} & \mu>\mu_0\\
2\,KL[\mu:\mu_0] & \mbox{if} & \mu\le\mu_0,
\end{array}
 \ea
 and $KL[\mu:\mu_0]=\log(\mu_0/\mu)-(\mu_0-\mu)/\mu_0$.  The intrinsic
priors are given in the next lemma (the proof is straightforward and
is omitted):
\begin{lem} The arithmetic and fractional intrinsic priors are
$$
\pi^A(\mu)={\mu_0}^{-1} (1+\frac{\mu}{\mu_0})^{-2},\hspace{.5cm}
\pi^F(\mu)={\mu_0}^{-1}\exp\{-\frac{\mu}{\mu_0} \}= Exp \, \{\mu \mid \frac{1}{\mu_0}\}.
$$
\end{lem}

\begin{figure}[!t]
\begin{center}
\begin{tabular}{cc}
\includegraphics[width=175pt,height=80pt]{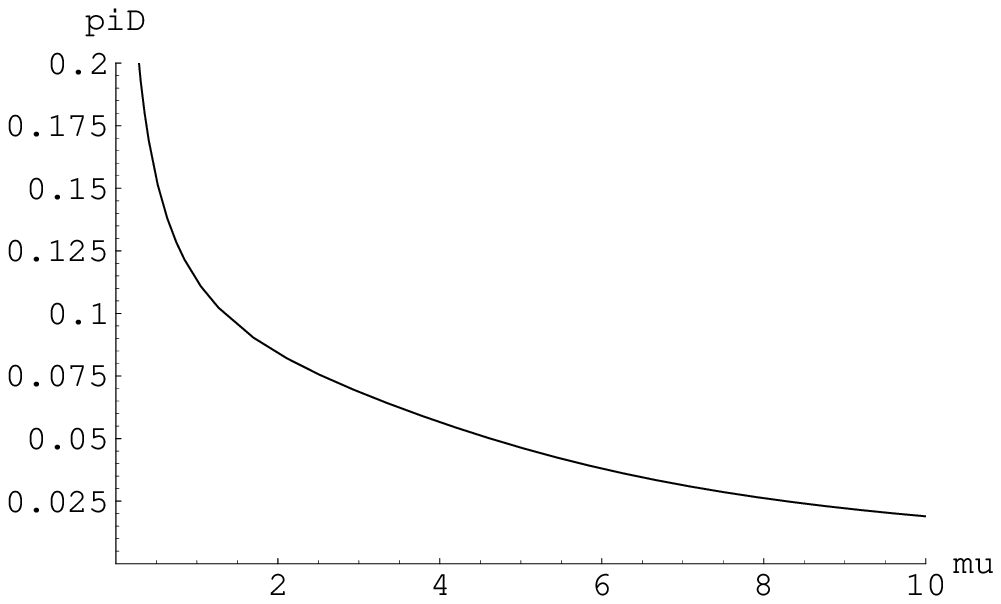} &
\includegraphics[width=175pt,height=80pt]{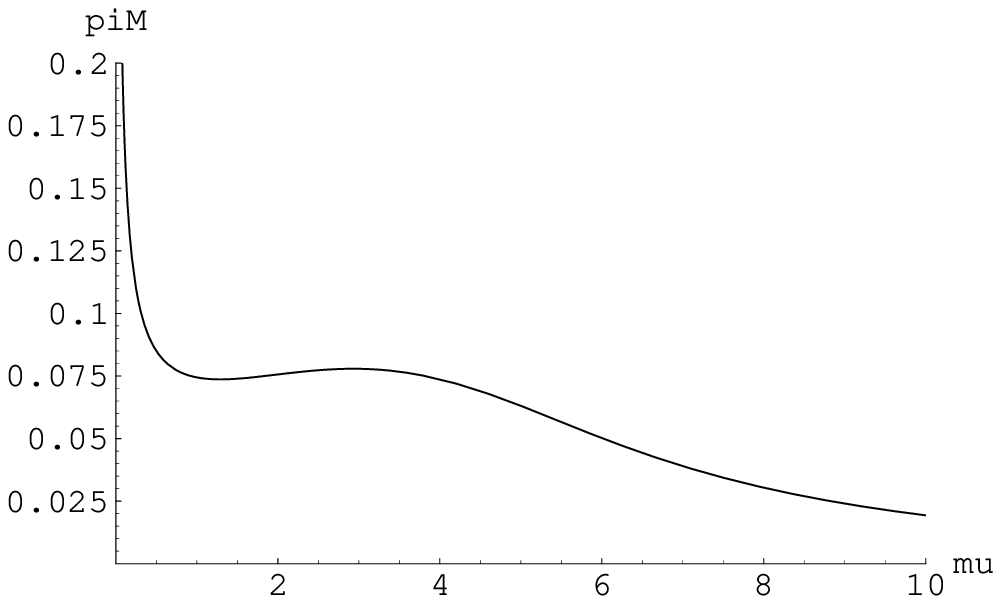}\\
\includegraphics[width=175pt,height=80pt]{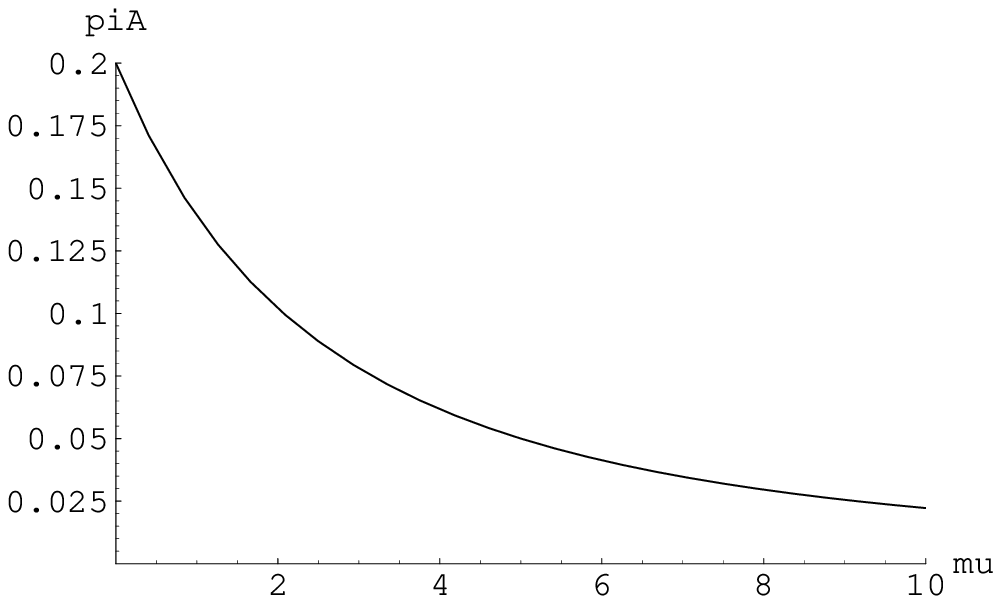} &
\includegraphics[width=175pt,height=80pt]{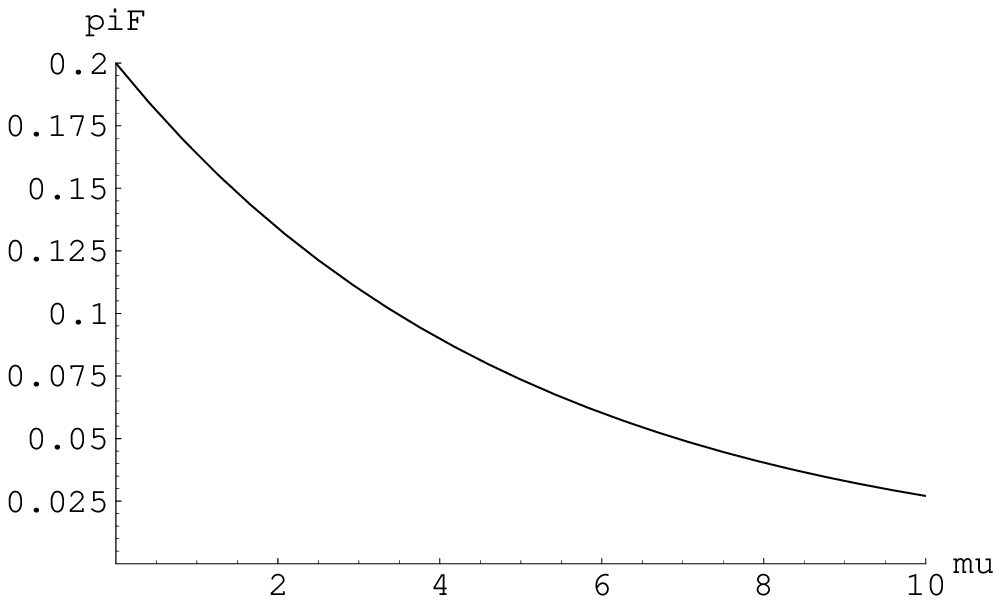}
\end{tabular}
\caption{$\pi^{S}$ (upper left), $\pi^{M}$ (upper right), $\pi^A$
(lower left) and $\pi^F$ (lower right) for the Exponential testing
of $\mu_0=5$.} \label{DBExp}
\end{center}
\end{figure}

The four priors are shown in Figure~\ref{DBExp} when testing $\mu_0=5$. They all have similar shapes,
although that of $\pi^M$ is somehow innusual;  they have some interesting properties:
\begin{enumerate}
    \item In the log scale, both $\pi^{M}$ and $\pi^{S}$ are symmetric around $\log\mu_0$; this is in
    accordance to Berger and Delampady (1987) and Berger and Sellke (1987) proposals, since $\log(\mu)$ is a location parameter.
    \item All four priors are proper.
    \item Neither the arithmetic intrinsic nor the DB priors have moments; the arithmetic fractional has all the moments.
    \item $\pi^{M}$ has the heaviest tails, and $\pi^F$ the thinnest. $\pi^{S}$ has heavier tails than $\pi^{A}$
    \item All four priors are `centered' at the null value $\mu_0$; indeed,  $\mu_0$ is the median of the DB priors
     and of $\pi^A$, and it is the mean of  $\pi^F$.
\end{enumerate}

The four Bayes factors $B_{12}$ in favour of $M_1: \mu = 5$ appear in Table~\ref{TabExp}, for two
values of $n$ ($n=10$ and $n=100$) and some few values of the  MLE $\hat\mu= \sum_{i=1}^n
y_i/n\in\{5,7.5,2.5\}$. We again find very similar results for the different priors,  with $B_{12}^S$
and $B_{12}^{A}$ providing slightly more support to $M_1$ than $B_{12}^M$ and $B_{12}^F$ when data is
compatible with $M_1$.

\begin{table}[t!]
\begin{center}{\small
\begin{tabular}{rc|ccccc}
$n=10$ & $\hat\mu$ & $B_{12}^S$ & $B_{12}^M$ & $B_{12}^A$ & $B_{12}^F$ \\
\hline
&5   & 5.65 &  4.43 & 5.13 & 3.59\\
&7.5 & 2.36 &  2.02 & 2.09 & 1.58\\
&2.5 & 0.95 &  0.88 & 0.82 & 0.59\\
$n=100$ & & & &\\
&5   & 17.28 &  12.81 & 15.98 & 10.89\\
&7.5 & $14.6\times10^{-4}$  & $12.2\times10^{-4}$  & $13\times10^{-4}$ & $9.4\times10^{-4}$\\
&2.5  & $0.86\times10^{-7}$ & $0.83\times10^{-7}$  & $0.73\times10^{-7}$ & $0.54\times10^{-7}$\\
\hline
\end{tabular}\caption{Bayes factors for the exponential testing with $\mu_0=5$ for different values of the
MLE and $n=10$, $n=100$.}\label{TabExp}}
\end{center}
\end{table}

We next investigate a desirable property of Bayes factors which
often fails when they are computed using conjugate priors (see
Berger and Pericchi, 2001). It is natural to expect that, for any
given sample size, $B_{12}\rightarrow 0$ as the evidence against the
simpler model $M_1$ becomes overwhelming. When this property holds,
we say that the Bayes factor is {\it evidence consistent} (or {\it
finite sample} consistent). It is easy to show that, if
$\bar{y}\rightarrow\infty$ then $B_{12}\rightarrow 0$ $\forall n$,
no matter what prior is used to obtain the Bayes factor. The
following lemma provides sufficient conditions for
$B_{12}\rightarrow 0$ as $\bar{y}\rightarrow 0$.

\begin{lem}\label{expCon} Let $B_{12}^\pi$ be the Bayes factor computed with $\pi(\mu)$.
$B_{12}^\pi\rightarrow 0$ as $\bar{y}\rightarrow 0$, for all $n\ge k>0$ if and only if
\begin{equation}\label{condExp}
\int_0^1 \mu^{-k}\, \pi(\mu)\, d\mu=\infty.
\end{equation}
\end{lem}
\begin{proof} See Appendix.\end{proof}
It follows that all four priors considered produce \emph{evidence consistent} Bayes factors for
all $n\ge 1$. Evidence consistency provides further insight into the behaviour of the DB priors.
Indeed, we recall that in the general definition of DB priors we used the power $\uq+\delta$, and
then we recommended the specific choice $\delta = .5$ . Interestingly, if $\delta>1$ is used
instead, then $\pi^{S}$ would not be evidence consistent as $\bar{y}\rightarrow 0$.

Last, we study the behavior of $B_{12}$ as the evidence in favor of $M_1$ grows (that is, as
$\bar{y}\rightarrow\mu_0$). For this example, it is easy to show that, when
$\bar{y}\rightarrow\mu_0$, $B_{12}$ grows to a constant, $B_{12}^0(n,\pi)$ say, that depends only
on $n$ and the prior used. Of course, it then follows from the dominated convergence theorem that
$B_{12}^0(n,\pi)\rightarrow\infty$ with $n$, but this also follows from general consistency of
Bayes factors (for proper, fix priors), so it is not very interesting. Of more interest for our
comparison is to study how fast $B_{12}^0(n,\pi_2)$ goes to $\infty$. In Figure~\ref{B0Exp} we
show $B_{12}^0(n,\pi)$ for the four priors considered. It can be seen that $\pi^{S}$ is the one
producing the largest values of $B_{12}^0$ for all values of $n$, with those for $\pi^{A}$
following very closely.

\begin{figure}[!t]
\centering
\includegraphics[width=190pt,height=120pt]{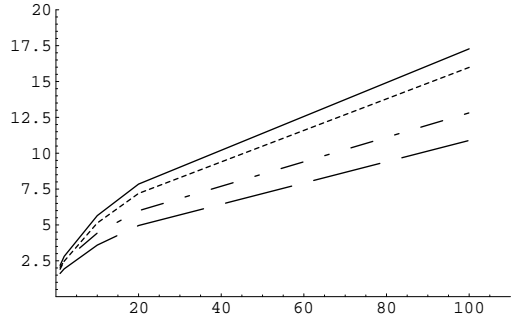}
\caption{Upper bounds $B_{12}^0(n,\pi)$ of Bayes factors as a function of $n$  for the priors
$\pi^{S}$ (Solid line), $\pi^{M}$ (Dot-dashed line), $\pi^F$ (Dashed Line), and $\pi^A$ (Dots).}
\label{B0Exp}
\end{figure}

\subsection{Location-scale (Example 3)}

DB priors are defined in general for vector parameters $\n\theta$.
As an illustration, we next consider a most popular example, namely
the normal distribution;  here the 2-dimensional $\n\theta$ has two
components of different nature (location and scale). Specifically,
assume that
$$
f(y\mid\mu,\sigma)=N(y\mid\mu,\sigma^2),
$$
and that we want to test $M_1:(\mu,\sigma)=(\mu_0,\sigma_0)$ versus
$M_2:(\mu,\sigma)\ne(\mu_0,\sigma_0)$. This hypothesis testing
problem occurs often in \emph{statistical process control}, where a
production process is considered `in control' if its production
outputs have a specified mean and standard deviation (the so called
{\it nominal values}); the question of interest is whether the
process is in control, that is, whether the mean and variance are
equal to the nominal values.

To compute the DB priors we use the reference prior $\pi^N(\mu,\sigma)=\sigma^{-1}$; for the
sum-DB prior we get:
$$
\pi^{S}(\mu,\sigma)=\pi^{S}(\sigma)\, \pi^{S}(\mu\mid\sigma), \hspace{.5cm}
\pi^{S}(\sigma)\propto\frac{\sigma} {(\sigma_0^4+\sigma^4)^{1/2}(\sigma_0^2+\sigma^2)^{1/2}},
$$
and
$$ \pi^{S}(\mu\mid\sigma)=\mbox{Ca}(\mu\mid \mu_0,\Sigma),\hspace{1cm}
\Sigma=\frac{\sigma_0^4+\sigma^4}{\sigma_0^2+\sigma^2},
$$
where $Ca$ represents the Cauchy density.
In this example, the minimum-DB prior $\pi^M$ does not exist, since
$\uq^M=\infty$. It can be checked that $\pi^{S}(\mu\mid\sigma)$ is
symmetric around $\mu_0$, which is a location parameter in
$\pi^{S}(\mu\mid\sigma)$; $\sigma_0$ is a scale parameter in
$\pi^{S}(\sigma)$. The joint density $\pi^{S}$ is shown in
Figure~\ref{DBNor}.

\begin{figure}[!t]
\centering
\includegraphics[width=190pt,height=120pt]{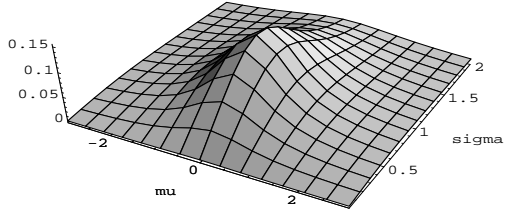}
\caption{$\pi^{S}$ for the Normal problem, with $\mu_0=0,\sigma_0=1$.} \label{DBNor}
\end{figure}

The intrinsic priors, which have simpler forms and thinner tails,
are derived next (the proof is omitted):

\begin{lem} The arithmetic intrinsic prior is
$$
\pi^A(\mu,\sigma)=\pi^A(\sigma)\, \pi^A(\mu\mid\sigma),\hspace{.5cm}
\pi^A(\sigma)=\frac{2}{\pi}\frac{\sigma_0}{\sigma^2+\sigma_0^2}, \hspace{.5cm}
\pi^A(\mu\mid\sigma)=N(\mu\mid\mu_0,\frac{\sigma^2+\sigma_0^2}{2}),
$$
and the fractional intrinsic prior is
$$
\pi^F(\mu,\sigma)=N^+(\sigma\mid 0, \frac{\sigma_0^2}{2}) \  N(\mu\mid\mu_0,\frac{\sigma_0^2}{2}),
$$
where $N^+$ stands for the normal density truncated to the positive real line.
\end{lem}

The intrinsic priors are proper; also, as with the sum-DB prior,  $\mu_0$ and $\sigma_0$ are location
and scale parameters for $\mu\mid \sigma$ and $\sigma$ respectively. Under the fractional intrinsic
prior $\pi^F$, $\mu$ and $\sigma$ are independent a priori.

Values of $B_{12}$  for all three priors and different values of the sufficient statistic
($\bar{y}$, $S$) are given in Table~\ref{LSTable} when  $(\mu_0,\sigma_0)=(0,1)$. The Bayes
factors corresponding to the different priors can be seen to be quite similar, specially, once
again, $B_{12}^S$ and $B_{12}^A$.

For the three priors, we display in Figure~\ref{margMP}  the
marginal distributions of $\sigma$ and in Figure~\ref{condMP}, the
conditional distributions of $\mu$ given $\sigma$. It can clearly be
seen
 that $\pi^F(\sigma)$  has thinner tails than  $\pi_2^A$ and $\pi_2^{S}$
(recall, thicker tails seem to perform better for testing). Also,
all conditional priors for $\mu$ are symmetric around their mode
$\mu_0$, with  $\pi^{S}(\mu\mid\sigma)$ having the heaviest tails.

\begin{table}[t!]
\begin{center}{\footnotesize
\begin{tabular}{c|ccc|ccc|ccc}
& \multicolumn{3}{c}{$\bar{y}=0$} & \multicolumn{3}{c}{$\bar{y}=1$} &
\multicolumn{3}{c}{$\bar{y}=2$}\\
\hline
& $B_{12}^s$ & $B_{12}^A$ & $B_{12}^F$ & $B_{12}^s$ & $B_{12}^A$ & $B_{12}^F$ & $B_{12}^s$ & $B_{12}^A$ & $B_{12}^F$ \\
$S=0.5$ & 2.30 & 1.35 & 0.70 & 0.03 & 0.02 & 0.01 & $3\cdot 10^{-8}$ & $4\cdot 10^{-8}$ &
$6\cdot 10^{-8}$\\
$S=1$ & 18.67 & 18.55 & 11.72 & 0.21 & 0.19 & 0.18 & $1\cdot 10^{-7}$ & $2\cdot 10^{-7}$ &
$6\cdot 10^{-7}$\\
$S=2$ & 0.006 & 0.006 & 0.017 & $5\cdot 10^{-5}$ & $5\cdot 10^{-5}$ & $21\cdot 10^{-5}$ &
$2\cdot 10^{-11}$ & $2\cdot 10^{-11}$ & $41\cdot 10^{-11}$\\
\hline
\end{tabular}\caption{For multidimensional parameter problem $(\mu_0=0,\sigma_0=1)$,
values of $B_{12}$ for different values of ($\bar{y},S$) with $n=10$.}\label{LSTable}}
\end{center}
\end{table}

\begin{figure}[!t]
\centering
\includegraphics[width=190pt,height=120pt]{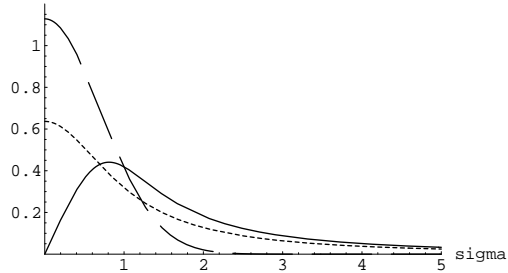}
\caption{Marginal distributions of $\sigma$ when $(\mu_0,\sigma_0)=(0,1)$; $\pi_2^{S}(\sigma)$
(solid line), $\pi_2^A(\sigma)$ (dots), and  $\pi_2^F(\sigma)$ (dashed line). The pair
(mode,median) for these priors are (0.81,1.56) for $\pi^D$, (0,1) for $\pi^A$, and (0,0.48) for
$\pi^F$.} \label{margMP}
\end{figure}

\begin{figure}[!t]
\centering
\begin{tabular}{cc}
\includegraphics[width=190pt,height=120pt]{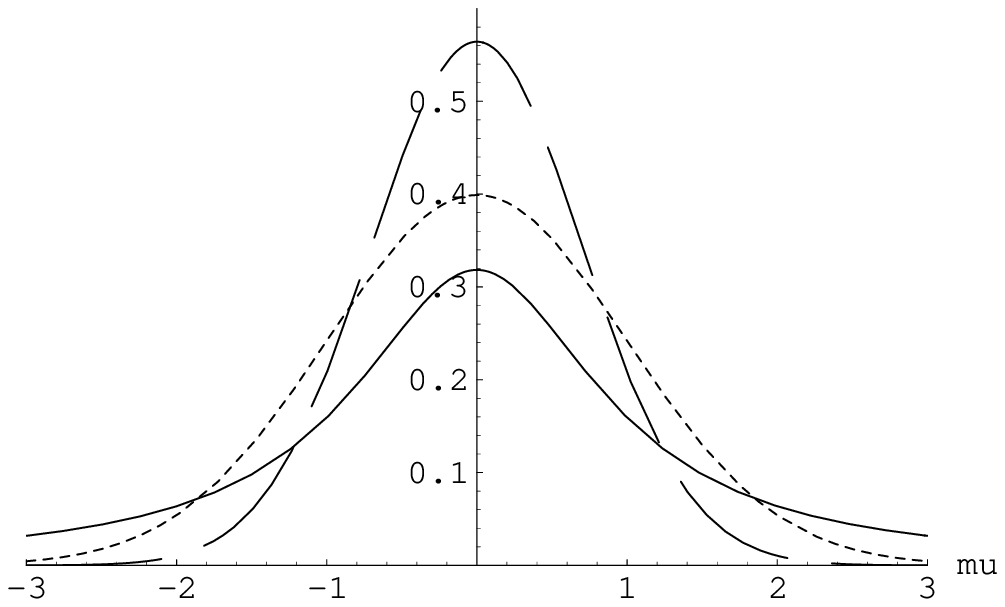} &
\includegraphics[width=190pt,height=120pt]{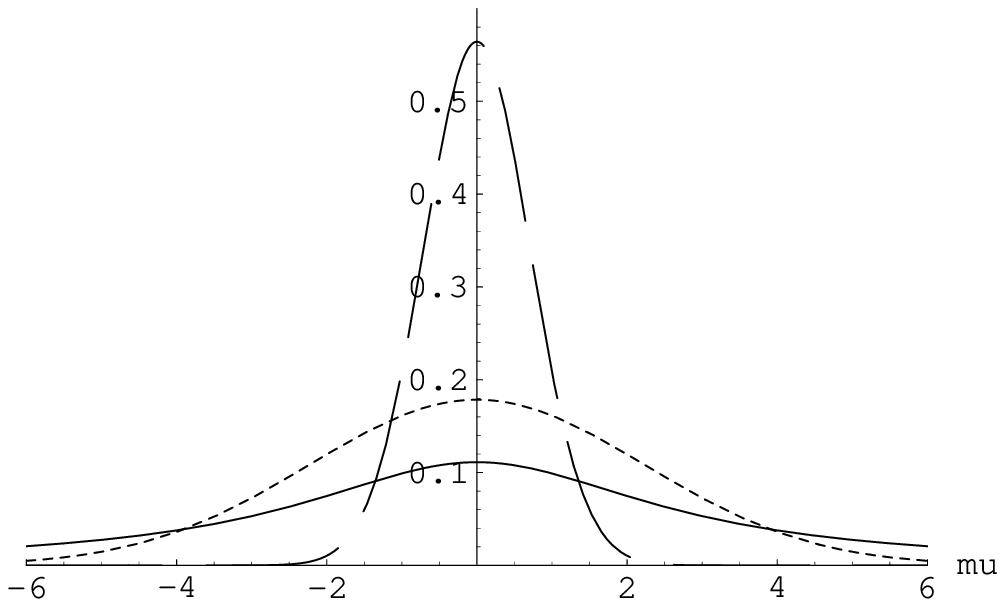}
\end{tabular}
\caption{ Conditional distributions of $\mu$ given $\sigma=1$ (left)
and $\sigma=3$ (right) when $(\mu_0,\sigma_0)=(0,1)$; $\pi^{S}$
(solid),  $\pi^A$ (dots), and  $\pi^F$ (dashed).}\label{condMP}
\end{figure}


With respect to the evidence consistency of the Bayes factors, it is easy to show that when either
$\bar{y}\rightarrow\infty$, $\bar{y}\rightarrow-\infty$ or $S\rightarrow\infty$ (the evidence
against $M_1$ is very strong), then $B_{12}\rightarrow 0$, $\forall n$ and for the three priors
considered. When the evidence in favor of  $M_1$ is largest (that is, $(\bar{y},S)\rightarrow
(\mu_0,\sigma_0)$) it can be seen (with a change of variables) that the Bayes factor in favor
of $M_1$, grows to $B_{12}^1(n,\pi)$
$$
B_{12}^1(n,\pi)=\int \beta^{-n}\exp\{-n  \,  \frac{1+\beta^2(\alpha^2-1)}{2\beta^2}\}\,
\pi_*^{j}(\alpha,\beta),
$$
a function only of $n$ and the prior ($j=A, F, S$) used. For the arithmetic intrinsic prior and fractional priors, the mixing densities $\pi_*^{j}$ are:
$$
\pi_*^A(\alpha,\beta)=\frac{2\beta}{\pi^{3/2}(1+\beta^2)^{3/2}}\exp\{-\frac{\alpha^2\beta^2}{1+\beta^2}\},\hspace{.5cm}
\pi_*^F(\alpha,\beta)=\frac{2\beta}{\pi}\exp\{-\beta^2(1+\alpha^2)\},
$$
and for the sum DB prior:
$$
\pi_*^S(\alpha,\beta)=\frac{\beta^2}{\pi\kappa}\big(1+\beta^4+\beta^2\alpha^2(1+\beta^2)\big)^{-1},\hspace{.5cm}
\kappa=\int s((1+s^4)(1+s^2))^{-1/2}\, ds.
$$
Figure \ref{B1NMu} illustrates the rate at which $B_{12}^1(n,\pi)\rightarrow\infty$ as $n\rightarrow\infty$. It can be clearly seen that, as in the previous example, DB and
intrinsic prior behave very similarly, being more sensitive to the evidence in favor $M_1$
than the fractional prior, subtantially so unless $n$ is very small.


\begin{figure}[!t]
\centering
\includegraphics[width=190pt,height=120pt]{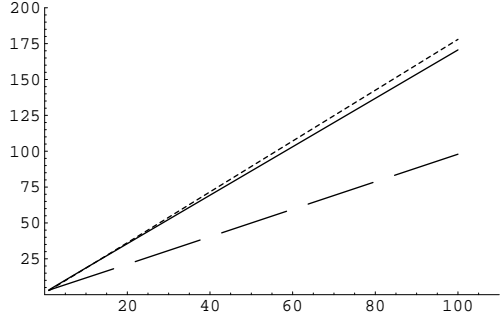}
\caption{Upper bounds $B_{12}^1(n,\pi)$ of Bayes factors as a function of $n$  for the priors
$\pi^{S}$ (Solid line), $\pi^{M}$ (Dot-dashed line), $\pi^F$ (Dashed Line), and $\pi^A$ (Dots).}
\label{B1NMu}
\end{figure}

Finally we compare the behavior of  the three priors in a real
example taken from Montgomery (2001). The example refers to
controlling the piston ring for an automotive engine production
process. The process was considered to be \emph{in control} if the
mean and the standard deviation of the inside diameter (in
millimeters) of the pistons were $\mu_0=74.001$ and
$\sigma_0=0.0099$. At some specific time, the following sample was
taken from the process:
$$
74.035,74.010,74.012,74.015,74.026\, ,
$$
and it had to be checked whether the process was in control. Bayes factors are given in
Table~\ref{NMuSig}.  $B_{12}^F$ provides about twice more support to $M_1$ than $B_{12}^S$ and
$B_{12}^A$, which are very similar to each other.

\begin{table}[h!]
\begin{center}{\small
\begin{tabular}{ccc}
$B_{12}^S$ & $B_{12}^A$ & $B_{12}^F$ \\
\hline
0.004 & 0.005 & 0.011\\
\hline
\end{tabular}\caption{Bayes factors $B_{12}$ for Montgomery (2001) example.}
\label{NMuSig}}
\end{center}
\end{table}

\subsection{Irregular models (Example 4)}
There is an important class of models for which the parameter space is constrained by the data. These
models do not have regular asymptotics and hence solutions based on asymptotic theory (like the
Bayesian information criteria, BIC) do not apply. Moreover, these models are very challenging for the
intrinsic approach; indeed, as discussed in  Berger and Pericchi (2001), the fractional Bayes factor
is completely unreasonable (and hence the fractional intrinsic prior is useless), and the arithmetic
intrinsic prior (which was only derived for the one side problem) is ``something of a conjecture''
(authors' verbatim). We take here the simplest such models, namely an exponential distribution with
unknown location. Accordingly, assume that
$$
f(y\mid\theta)=\exp\{-(y-\theta)\},\hspace{1cm}y>\theta,
$$
and that it is wanted to test $H_1:\theta=\theta_0$ vs.  $H_2:\theta\ne\theta_0$. To the best of our
knowledge, no objective priors have been proposed for this testing problem in the literature.

In these situations, the sum-symmetrized kulback-Leibler divergence $D^S[\theta, \theta_0]$ is
$\infty$, so we have to use the minimum. It can be checked that
$\bar{D}^M[\theta,\theta_0]=2|\theta-\theta_0|$, a well defined divergence. Also,
$\pi^N(\theta)=1$ since $\theta$ is a location parameter. The Minimum DB prior is then given by
$$
\pi^{M}(\theta)= \frac{1}{2}\big(1+2|\theta-\theta_0|\big)^{-3/2},\hspace{1cm} \theta\in{\cal R},
$$
which is symmetric with respect to $\theta_0$ (as expected, since  $\theta$ is a location
parameter); also, $\pi^{M}$ has no moments. Figure~\ref{DBIrr} (left) shows  $\pi^{M}(\theta)$
when $\theta_0 = 0$.

We next investigate the {\it evidence} consistency for any $n$.  The
sufficient statistic is $T=\min\{y_1,\ldots,y_n\}$. It is trivially
true that $B_{12}\rightarrow 0$, as $T\rightarrow-\infty$  for any
(proper) prior (in fact, $B_{12}=0$ for $T < \theta_0$). The next lemma provides a sufficient condition on
the prior to produce evidence consistency  $\forall n$, as
$T\rightarrow\infty$.

\begin{lem}\label{ecIM} Let $\pi(\theta)$ be any proper prior (on $M_2$) and $B_{12}^\pi$ be the corresponding
Bayes factor. If for some integer $k>0$
 \beq{irrcs}
  \int_{\theta_0}^\infty e^{k\theta}\, \pi(\theta)\,
d\theta=\infty,
 \eeq
 then $B_{12}^\pi\rightarrow 0$ as $T\rightarrow \infty$
$\forall n\ge k$.
\end{lem}
\begin{proof} See Appendix.\end{proof}
It follows from the previous lemma that $\pi^{M}$ produces \emph{evidence consistent} Bayes
factors $\forall n\ge 1$. We next investigate the situation for increasing evidence {\it in favor}
of $M_1$, that is, as $T\rightarrow \theta_0^+$. Let

\begin{figure}[!t]
\centering
\begin{tabular}{cc}
\includegraphics[width=190pt,height=120pt]{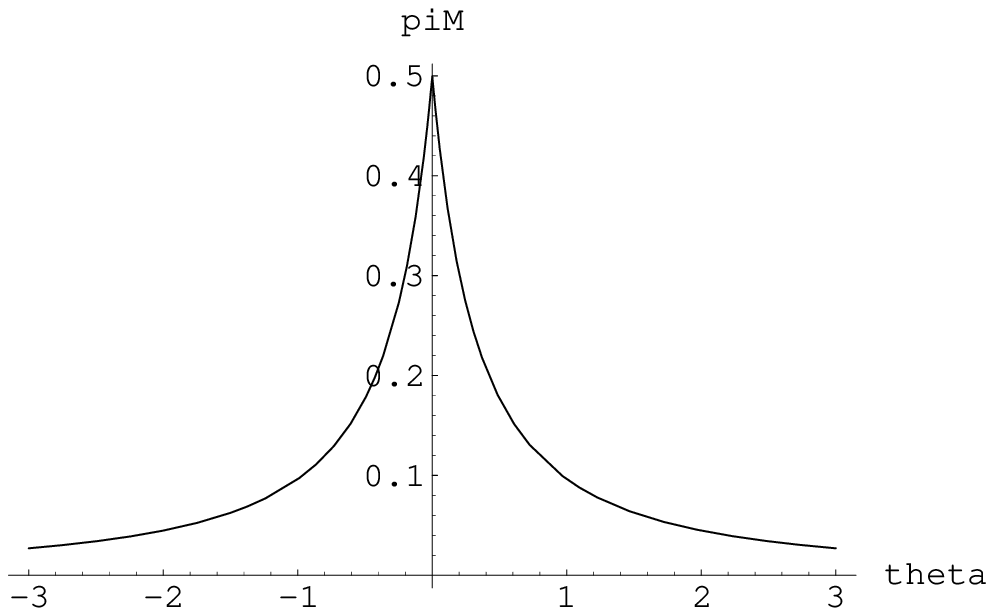} &
\includegraphics[width=190pt,height=120pt]{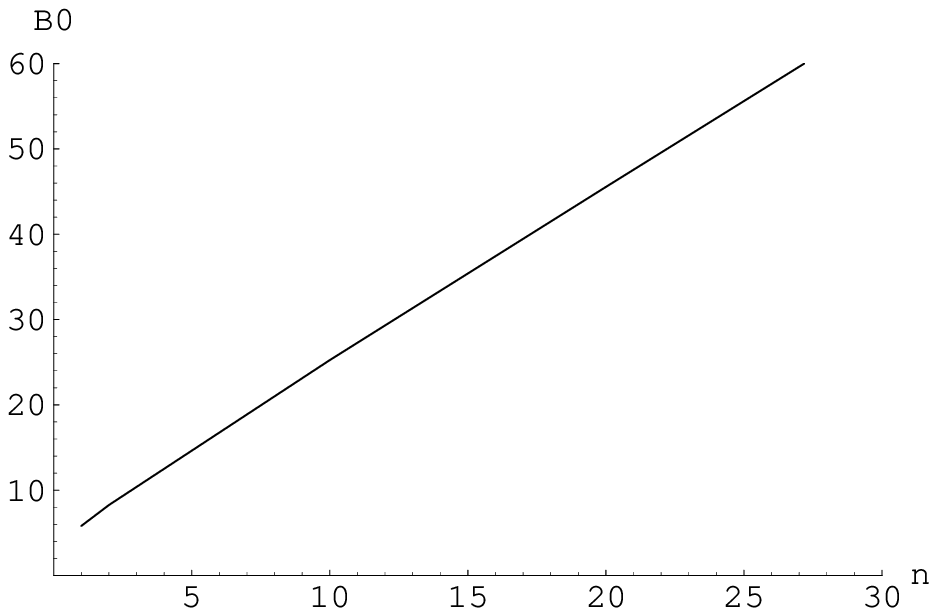}
\end{tabular}
\caption{Irregular example, two-side testing of $M_1: \theta=0$. Left: the DB prior $\pi^{M}$; Right:
$B_{12}^0(n)$ as a function of $n$.} \label{DBIrr}
\end{figure}

$$
B_{12}^0(n)=\lim_{T\rightarrow \theta_0^+}\, B_{12}^{\pi^D}.
$$
$B_{12}^0(n)$ is an upper bound of $B_{12}$ when the evidence in favor of $M_1$ is largest. It can be
seen in Figure~\ref{DBIrr} (right) that  $B_{12}^0(n)$ is nearly linear. Of course
$B_{12}^0(n)\rightarrow\infty$ when $n\rightarrow\infty$.

As mentioned before, there does not seem to be any other proposals in the literature for the two-side
testing problem. However, Berger and Pericchi (2001), do consider the `one side testing' version,
namely testing $M_1:\theta=\theta_0$ vs $M_2:\theta>\theta_0$; they conjecture that the arithmetic
intrinsic prior for this problem is the proper density
$$
\pi_2^A(\theta)=\big(-e^{\theta-\theta_0}\,\log(1-e^{\theta_0-\theta})-1\big),\hspace{.5cm}\theta>\theta_0,
$$
which is a decreasing and unbounded function of $\theta$. Also, since We next compare the (minimum) DB prior for
this problem with Berger and Pericchi proposal.

Although our original formulation appears to be in terms of two side testing (see \eqref{ht}) in
reality it suffices to define $\Theta$ appropriately to cover other testing situations. For
instance, in our one-side testing, we take $\Theta = [\theta_0, \infty)$. The (minimum) DB prior
is
$$
\pi^{M}(\theta)=\big(1+2(\theta-\theta_0)\big)^{-3/2},\hspace{.5cm}\theta>\theta_0.
$$

It can  be checked, that $\pi^A$ meets condition (\ref{irrcs}) for $k=1$ and  hence $\pi^A$ produces
evidence consistent Bayes factors $\forall n\ge 1$. The priors  $\pi^A$ and $\pi^{M}$ are displayed
in Figure~\ref{Priors1S}. We find that also in this example $\pi^{M}$ has thicker tails.

In this one side testing  scenario (in sharp contrast to the
behavior in the two-side testing) the Bayes factor in favor of $M_1$
for every $n > 0$ does grow to $\infty$ as the evidence in favor of
$M_1$ grows. Indeed, the Bayes factor $B_{12}$ is
$$
\Big(\int_{\theta_0}^T\, \exp\{n(\theta-\theta_0)\}\, \pi(\theta)d\theta\Big)^{-1},
$$
so that, $B_{12}\rightarrow\infty $ when $T\rightarrow\theta_0^+ \ ,
\forall n > 0$, no matter what prior is used. Note that here
$\theta_0$ is in the boundary of the parameter space.

\begin{figure}[!t]
\centering
\includegraphics[width=200pt,height=120pt]{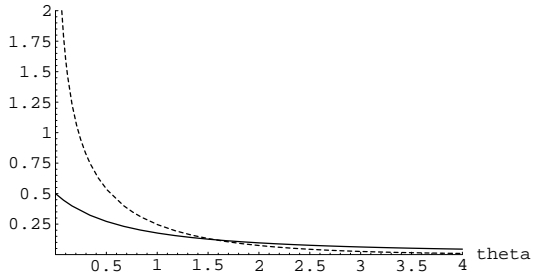}
\caption{Irregular, one side testing problem: $\pi^D$ (solid) and $\pi^A$ (dots) for the case
$\theta_0=0$.} \label{Priors1S}
\end{figure}

In Table~\ref{IrrTable}, we produce the Bayes factors computed with
$\pi^A$ and $\pi^{M}$ when $\theta_0=0$ for various values of
$T=\min\{y_1,\ldots,y_n\}$, and for $n=10$ and $n=20$. For small
values of $T$ ($T<0.20$), when evidence supports $M_1$, $B_{12}^{M}$
is considerably larger than $B_{12}^{A}$, thus giving more support
to $M_1$. For larger values of $T$ (that is, when data contradict
$M_1$) both priors result in very similar Bayes factors.

\begin{table}[t!]
\begin{center}{\small
\begin{tabular}{c|cccccc}
& \multicolumn{6}{c}{$T$}\\
\hline
            & 0.02 & 0.05 & 0.10 & 0.20 & 0.50 & 1.00\\
& \multicolumn{6}{c}{$n=10$}\\
$B_{12}^M$ & 46.56 & 16.66 & 6.83 & 2.19 & 0.16 & 0.002\\
$B_{12}^A$     & 11.54 & 5.16 & 2.57 & 1.02 & 0.10 & 0.001\\
\hline
& \multicolumn{6}{c}{$n=20$}\\
$B_{12}^M$ & 41.96 & 12.65 & 3.75 & 0.55 & 0.002 & $2\cdot 10^{-7}$\\
$B_{12}^A$     & 10.52 & 4.04 & 1.50 & 0.28 & 0.002 & $2\cdot 10^{-7}$\\
\hline
\end{tabular}\caption{Irregular models, one side testing.  Values of $B_{12}$
for different values of $T$, $n$ and for the two priors $\pi^A,
\pi^M$, when testing $\theta_0=0$.}\label{IrrTable}}
\end{center}
\end{table}

\subsection{Mixture models (Example 5)}

Mixture models are among the most challenging scenarios for
objective Bayesian methodology. These models have {\it improper
likelihoods}, i.e., likelihoods for which no improper prior yields a
finite marginal density (integrated likelihood). Recently, P\'{e}rez
and Berger (2001), have used \emph{expected posterior priors} (see
P\'{e}rez and Berger, 2002) to derive objective estimation priors,
but basically no general method seems to exist for deriving
objective priors for testing with these models.

However, the divergence measures are well defined (although the integrals are now more involved)
providing a reasonable DB prior to be used in model selection. We consider a simple illustration.
Assume
$$
f(y\mid\mu,p)=p\, N(y\mid 0,1)+(1-p)\, N(y\mid \mu,1),
$$
and the testing of $H_1:\mu=0$, vs. $H_2:\mu\ne 0$, where $p<1$ is
known (if $p=1$, both hypotheses define the same model). As Berger
and Pericchi (2001) point out, there is no minimal training sample
for this problem and hence the intrinsic Bayes factor cannot be
defined. The fractional Bayes factor does not exist either. The only
prior we know for this problem is the recommendation in Berger and
Pericchi (2001) of using $\pi^{BP}(\mu)=Ca(\mu|0,1)$.

Although there is no formal $\pi^N(\mu)$ here, $\pi^N(\mu)=1$ is usually assumed  (see for
instance P\'{e}rez and Berger, 2002). It can be shown that $\uq^M=\infty$, and hence, $\pi^{M}$ does
not exist. Let
\begin{equation}\label{Greal}
G(p,\mu,\mu^*)=\int_{-\infty}^\infty\, \log\Big[1+\frac{1-p}{p}\, e^{y\mu-\mu^2/2}\Big]
N(y\mid\mu^*,1)\, dy.
\end{equation}
Then
$$
D^S[\mu,\mu_0]=n(1-p)\big(G(p,\mu,\mu)-\,G(p,\mu,0)\big).
$$
It can be shown that  $\uq^S<\infty$, and hence that the sum DB prior $\pi^{S}$ exists. The
normalizing constant, however, can not be derived in closed form. Numerical procedures could be used
to exactly derive the sum-DB prior. We use instead a Laplace approximation (see Tanner 1996) to
\eqref{Greal} to get and approximate DB prior. Specifically
%
%
\begin{equation}\label{Gaprox}
G(p,\mu,\mu^*)\approx\log\Big[1+\frac{1-p}{p}\, e^{\mu^*\mu-\mu^2/2}\Big]= G^L(p,\mu,\mu^*).
\end{equation}
Figure~\ref{Grealaprox} shows $G(p,\mu,\mu^*) - G(p,\mu,0)$ and its approximation
$G^L(p,\mu,\mu^*) - G^L(p,\mu,0)$ for $p=.5$ and $p=.75$. The approximation is very good as long
as $p$ is not too extreme.

\begin{figure}[!t]
\centering
\begin{tabular}{cc}
\includegraphics[width=190pt,height=120pt]{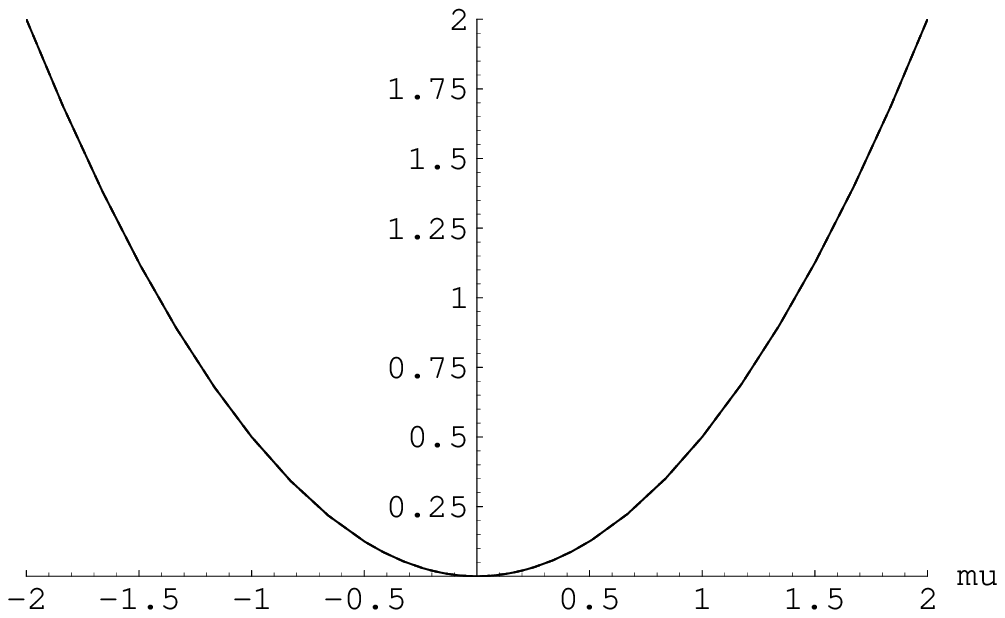} &
\includegraphics[width=190pt,height=120pt]{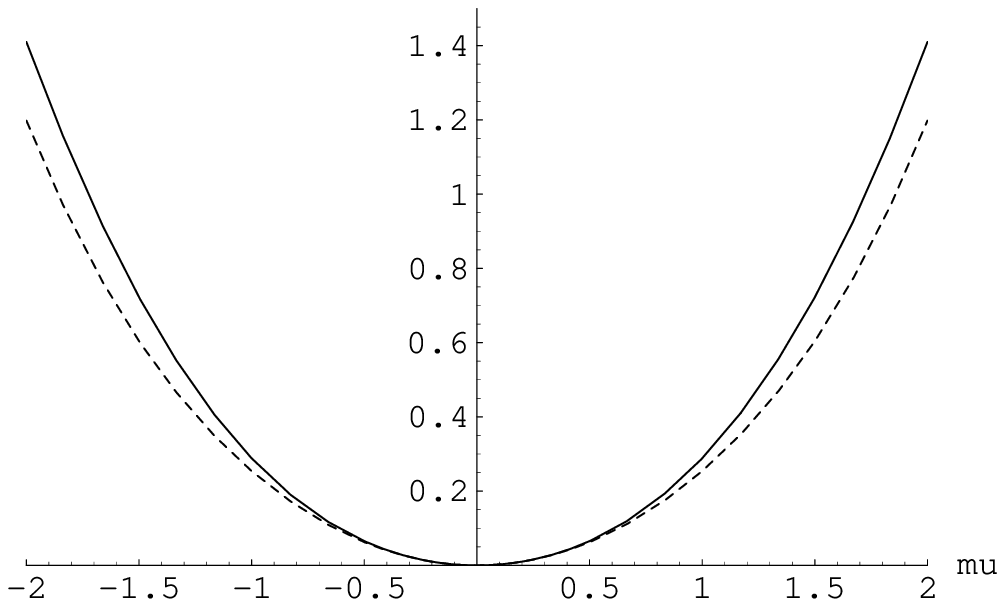}
\end{tabular}
\caption{$G(p,\mu,\mu)-G(p,\mu,0)$ (solid) and its Laplace approximation
$G^L(p,\mu,\mu)-G^L(p,\mu,0)$ (dots). Left: $p=0.50$. Right: $p=0.75$.}\label{Grealaprox}
\end{figure}

We can now use this approximation to derive the DB prior. Note that the natural effective sample size here is
$n^*=n(1-p)$, so that the unitary sum-symmetrized divergence is

$$
\bar{D}^S[\mu,\mu_0]=\frac{D^S(\mu,\mu_0)}{n(1-p)}\approx \log\frac{1+\frac{1-p}{p} \
e^{\mu^2/2}}{1+\frac{1-p}{p} \ e^{-\mu^2/2}}=\bar{D}^{S L}[\mu,\mu_0].
$$

This approximation is specially appealing because it also keeps essential properties of the
divergence measures. In particular, $\bar{D}^{S L}(\mu,\mu_0)\ge \bar{D}^{S L}(\mu_0,\mu_0)=0$, so
that the approximate DB prior
$$
\pi^{SL}(\mu)\propto \big(1+\bar{D}^{S L}(\mu,\mu_0)\big)^{-q_*^s},
$$
has a mode at zero. Since  $\uq^s=1/2$, we finally get
$$
\pi^{SL}(\mu)\propto \big(1+\bar{D}^{L S}(\mu,\mu_0)\big)^{-1}.
$$

\begin{figure}[!t]
\centering
\begin{tabular}{ccc}
\includegraphics[width=130pt,height=120pt]{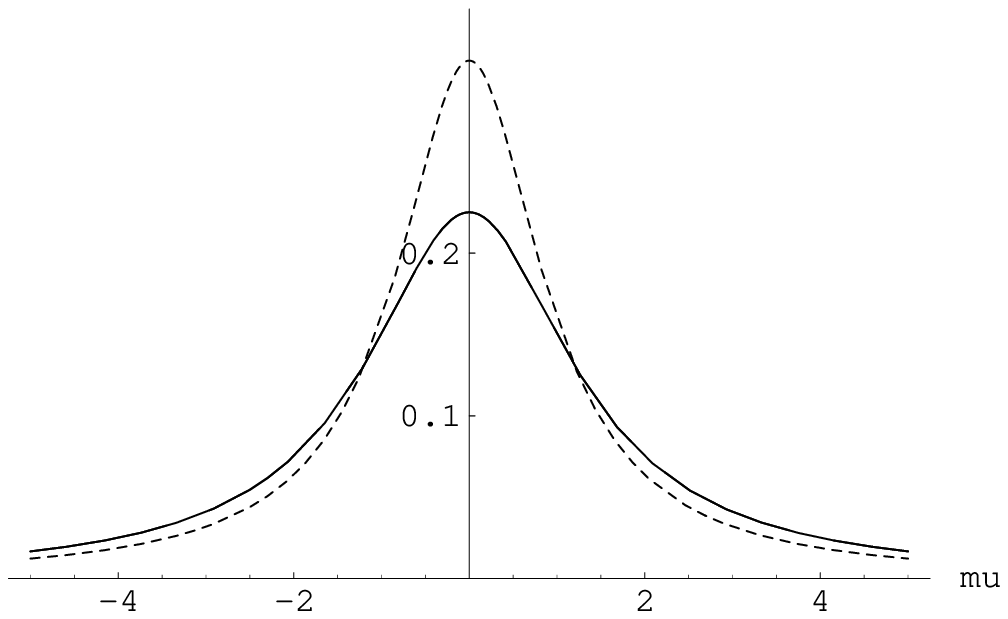} &
\includegraphics[width=130pt,height=120pt]{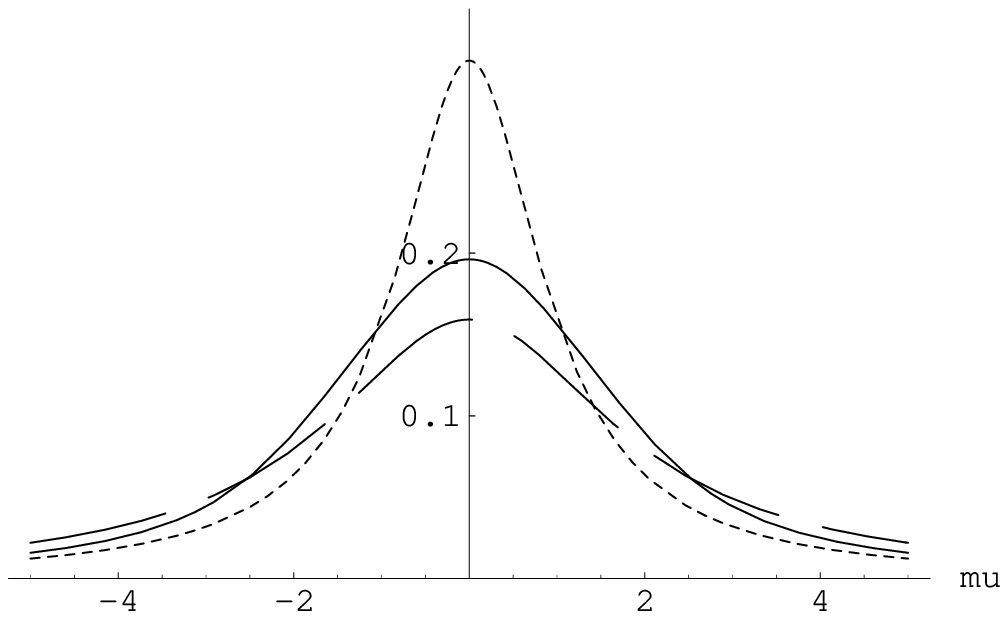} &
\includegraphics[width=130pt,height=120pt]{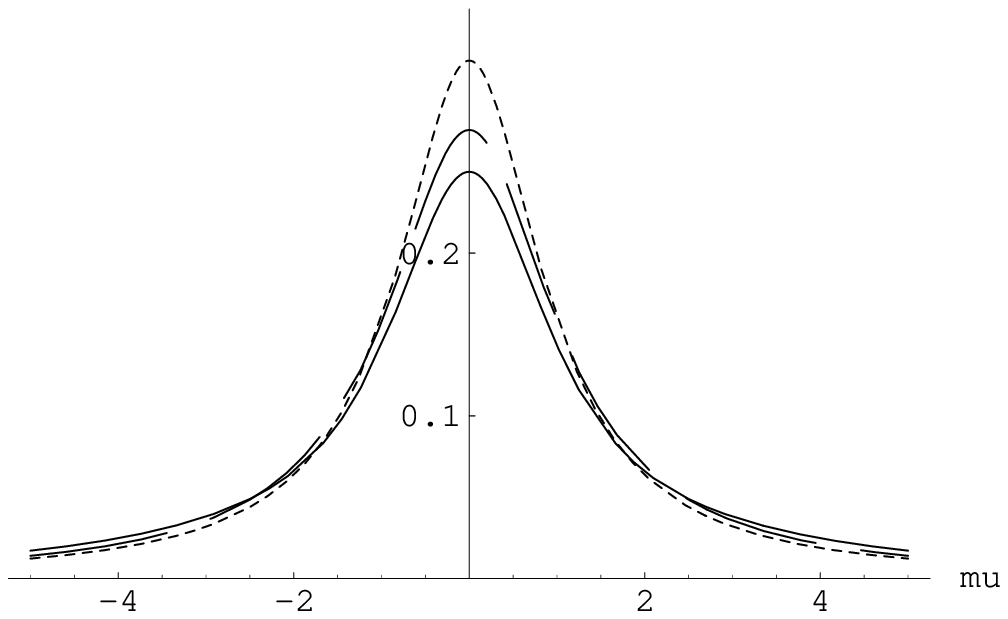}
\end{tabular}
\caption{$\pi^{SL}$ (solid line), $Ca(0,1/1-p)$ (dashed line) and $\pi^{BP}(\mu)=Ca(\mu|0,1)$
(dots) for $p=0.50$ (left), $p=0.75$ (middle) and $p=0.25$ (right).}\label{piCas}
\end{figure}

Interestingly, the prior $\pi^{SL}$ is close to a Cauchy density, which was Berger and Pericchi
proposal, although the scale differs. Indeed a Taylor expansion of order 3, around $\mu=0$ gives
\begin{equation}\label{firstap}
\bar{D}^{SL}(\mu,\mu_0) \approx (1-p)\mu^2,
\end{equation}
so that, unless $p$ is very close to 1, $\pi^{SL}$ behaves around
$0$ as a $Ca(\mu \mid 0,1/(1-p))$; the approximation is excellent
when $p$ is close to 0.5. In the tails, on the other hand, we have
that, as $|\mu|\rightarrow\infty$
\begin{equation}\label{secondap}
\bar{D}^{SL}(\mu,\mu_0) \approx \frac{\mu^2}{2},
\end{equation}
independently of $p$. Hence, the tails of $\pi^{SL}$ are close to those of a $Ca(\mu\mid 0,2)$
density. Note that both approximations (\ref{firstap}) and (\ref{secondap}) coincide for $p=0.5$.

The scale of the $Ca(\mu \mid 0,1/(1-p))$ makes intuitive sense.
Indeed, the larger $p$, the less observations providing information
about $\mu$ we get, and the DB prior adjust to a less informative
likelihood by inflating its scale.  Figure~\ref{piCas} displays
$\pi^{SL}$, its $Ca(\mu \mid 0,1/(1-p))$ approximation, and the
proposal of Berger and Pericchi (2001) { for different values of
$p$. Notice that, for values of $p$ close to 0, $\pi^{SL}$ (and its
approximation $Ca(0,1/1-p)$) approximately behaves as a $Ca(0,1)$,
the Berger and Pericchi proposal (see Figure~\ref{piCas}, right).
This has an interesting interpretation since, as $p\rightarrow 0$
the testing problem in this example essentially coincides with that
of testing $H_1:\mu=0$ vs. $H_2:\mu\ne0$, when $\mu$ is the mean of
a normal density, for which the $Ca(\mu\mid 0,1)$ is perhaps the
most popular prior to be used as prior distribution for $\mu$ under
$H_2$. }

In this example, the DB prior (as well as Berger and Pericchi proposal) again produces evidence
consistent Bayes factors for all $n$. Indeed, it can be shown that if one of the $y_i's$ tends to
$\infty$ or $-\infty$, then the corresponding Bayes factor tends to 0 no matter what prior is used.
On the other hand, as the evidence for $H_1$ increases, we get a finite upper bound on $B_{12}$ for
every fixed sample size $n$:
$$
B_{12}^0(n,p,\pi)=\lim_{y_i\rightarrow 0,\forall i}\, B_{12}.
$$
In Figure~\ref{ILM0} we show $B_{12}^0$ for $\pi=\pi^{SL}$ and $\pi=Ca(\mu \mid 0,1)$ as a
function of $n$ for $p=0.5$.  As in the previous examples, it is an immediate consequence that
$B_{12}^0(n,p,\pi)\rightarrow\infty$ as $n\rightarrow\infty$ for both priors, but the support for
$H_1$ is larger when $\pi^{SL}$ is used for every  $n$.

\begin{figure}[!t]
\centering
\includegraphics[width=190pt,height=120pt]{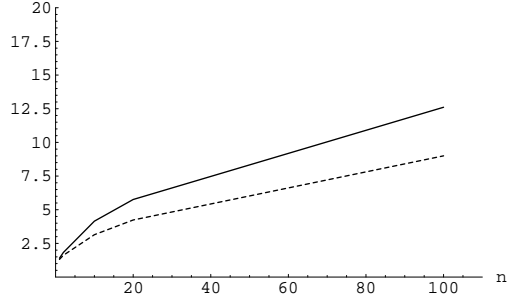}
\caption{$B_{12}^0$ for $\pi^{SL}$ (solid line) and $\pi^{BP}$ (dots) as a function of $n$, for
$p=0.5$.} \label{ILM0}
\end{figure}

In Table~\ref{ILTable} we show the Bayes factors $B_{12}^{SL}$,  $B_{12}^{ap}$ and $B_{12}^{BP}$
computed respectively with the priors $\pi^{SL}$, its $Ca(\mu\mid0,1/(1-p))$ approximation and the
$Ca(\mu\mid 0,1)$ proposed by Berger and Pericchi. Since reduction by sufficient statistic is not
possible, the Bayes factors are computed for simulated samples of size $n=20$, with mean
$\mu\in\{0,0.5,1\}$, and $p\in\{0.25,.5, 0.75\}$. $B_{12}^{SL}$ and its approximation
$B_{12}^{ap}$ are very close, demonstrating that the approximation is very good for the considered
range of $p$. $B_{12}^{SL}$ and $B_{12}^{BP}$ are also very similar.

\begin{table}[t!]
\begin{center}{\small
\begin{tabular}{c|ccc|ccc|ccc}
&\multicolumn{3}{c}{$p=0.25$} &\multicolumn{3}{c}{$p=0.5$} &\multicolumn{3}{c}{$p=0.75$} \\
$\mu$ & $B_{12}^{SL}$ & $B_{12}^{ap}$ & $B_{12}^{BP}$ & $B_{12}^{SL}$ & $B_{12}^{ap}$ &
$B_{12}^{BP}$
& $B_{12}^{SL}$ & $B_{12}^{ap}$ & $B_{12}^{BP}$\\
\hline
0 & 5.49 & 4.97 & 4.39 & 2.56 & 2.56 & 2.01 & 2.37 & 2.90 & 1.87\\
0.5 & 1.82 & 1.65 & 1.49 & 0.36 & 0.36 & 0.33 & 1.69 & 2.06 & 1.42\\
1 & 0.07 & 0.06 & 0.06 & 0.04 & 0.04 & 0.04 & 0.01 & 0.01 & 0.01\\
\hline
\end{tabular}\caption{Bayes factors $B_{12}$ for simulated samples of size $n=20$ from the mixture
model with various values of $p$ and $\mu$ and the priors
$\pi^{SL}$ , its approximation $Ca(\mu \mid 0,1/(1-p))$ and
$\pi^{BP}(\mu)= Ca(\mu \mid 0,1)$.}\label{ILTable}}
\end{center}
\end{table}

\section{Nuisance parameters}\label{nuisance}
In this section we deal with more realistic problems in which the distribution of the data is not
fully specified under the null (simplest model), but depends on some nuisance parameter. Assume
that $y_i, \ i=1, \ldots, n$ are independent (not necessarily i.i.d.) and that $\n y = (y_1,
\ldots, y_n) \sim \{f(\n y \mid \n\theta, \n\nu), \ \n\theta\in\Theta, \, \n\nu\in\Upsilon\}$. We
want to test $H_1:\n\theta=\n\theta_0$ vs. $H_2:\n\theta\ne\n\theta_0$. Equivalently we want to
solve the model selection problem  \eqref{comp} where it is carefully acknowledged that $\n\nu$
can have a different meanings in each model.

However, from now on we assume, after suitable reparameterization if
needed, that $\n\theta$ and $\n\nu$ are {\it orthogonal} (that is,
that Fisher information matrix is block diagonal). It is then
customary to assume that $\n\nu$ has the same meaning under both
models (see Berger and Pericchi, 1996, for an asymptotic
justification). This will be needed for the divergence measures to
have intuitive meaning, and also to justify assessment of the same
(possibly improper) prior for $\n\nu$ under both models thus
considerably simplifying the assessment task. The suitability of
orthogonal parameters in the presence of model uncertainty was first
exploited by Jeffreys (1961) and it has been successfully used by
many others (see for example Zellner and Siow, 1980, 1984, and
Clyde, DeSimone and Parmigiani, 1996). For univariate $\theta$, Cox
and Reid (1987) explicitly provide an orthogonal reparameterization.

Accordingly, we assume that the hypothesis testing problem above is
equivalent to that of choosing between the competing models:
 \beq{compNuiO}
M_1:f_1(\n{y}\mid\n\nu)=f(\n{y}\mid\n\theta_0,\n\nu)\hspace{.25cm}\mbox{vs.} \hspace{.25cm}
M_2:f_2(\n{y}\mid\n\theta,\n\nu)=f(\n{y}\mid\n\theta,\n\nu),
 \eeq
  where $\n\theta_0\in\Theta$ is a specified value, and $\n\nu$ (the {\it old parameter} in Jeffrey's
  terminology) is assumed to be common to both models, which only differ by the different value of
  the {\it new parameter} $\n\theta$ under $M_2$.

\subsection{Divergence Measures}


The basic measure of discrepancy between $\n\theta$ and $\n\theta_0$ is again Kullback-Leibler
directed divergence \eqref{eq:KL} where $\n\nu$ is taken to be the same in both models:
\ba KL[(\n\theta_0,\n\nu):(\n\theta,\n\nu)]=\int_{{\cal Y}}\, (\log f(\n y\mid \n\theta, \n\nu)-\log
f(\n y\mid\n\theta_0,\n\nu))\, f(\n y\mid\n\theta, \n\nu)\,d\n y. \ea
Note that using the same $\n\nu$ only makes intuitive sense if $\n\nu$ has the same meaning under
both models, and hence can be considered common. Actually, P\'{e}rez (2005) using geometrical
arguments, shows that under orthogonality $ KL[(\n\theta_0,\n\nu):(\n\theta,\n\nu)]$ can be
interpreted as a measure of divergence between $f_1$ and $f_2$ due solely to the parameter of
interest $\n\theta$. This interpretation does not hold for other divergence measures,  as the
intrinsic loss divergence defined in Bernardo and Rueda (2002).

Similarly to Section~\ref{DBpriors} we  symmetrize Kullback-Leibler
directed divergence by adding or taking the minimum of them,
resulting in the sum-divergence and min-divergence measures between
$\n\theta$ and $\n\theta_0$ for a given $\n\nu$
 \beq{DSnu}
  D^S[(\n\theta,\n\theta_0)\mid \n\nu]=
KL[(\n\theta,\n\nu):(\n\theta_0,\n\nu)]+ KL[(\n\theta_0,\n\nu):(\n\theta,\n\nu)],
  \eeq
and
  \beq{DMnu}
D^M[(\n\theta,\n\theta_0)\mid \n\nu]= 2\times\;\min\{KL[(\n\theta,\n\nu):(\n\theta_0,\n\nu)],
KL[(\n\theta_0,\n\nu):(\n\theta,\n\nu)]\}.
  \eeq
  $D^M$ is used by P\'{e}rez (2005) {to define what he calls} the ``orthogonal intrinsic loss''.

In what follows, many of the definitions and properties apply to
both $D^S$ and $D^M$, in which case we again generically  use $D$ to
denote any of them. Their basic properties were discussed in
Section~\ref{DBpriors}. As before, the building block of the DB
prior is the {\it unitary measure of divergence} $\bar D = D/n^*$,
where  $n^*$ is the equivalent sample size  for $\n\theta$.

\subsection{DB priors in the presence of nuisance parameters}
For testing $H_1:\n\theta=\n\theta_0$ vs. $H_2:\n\theta\ne\n\theta_0$, or equivalently choosing
between models $M_1$ and $M_2$ in  \eqref{compNuiO}, we need priors $\pi_1(\n\nu)$ under $M_1$ and
$\pi_2(\n\nu, \n\theta)$ under $M_2$.

In the spirit of Jeffreys (and many others after him) we  take
(under each of the models) the {\it same} objective (possibly
improper) prior for the common parameter $\n\nu$ and a proper prior
for the conditional distribution of the new parameter $\n\theta \mid
\n\nu$ under $M_2$, which will be derived similarly to the DB priors
in Section~\ref{GenPar}. Note that since $\n\nu$ occurs in the two
models, if we take the same $\pi^N(\n\nu)$ in both, then the
(common) arbitrary constants cancel when computing the Bayes factor;
however $\n\theta$ which only occurs in $M_2$ has to have a proper
prior. A common prior for the old parameter only makes sense when
$\n\nu$ has the same meaning in both models (another reason to take
$\n\theta$ and $\n\nu$ orthogonal). Moreover, it is well known that
under orthogonality, the specific \emph{common} prior for $\n\nu$
has little impact on the resulting Bayes factor (see Jeffreys 1961;
Kass and Vaidyanathan 1992), thus supporting use of objective priors
for common parameters.


Let $\pi^N(\n\nu)$ be an objective (usually either Jeffreys or reference) prior for model $f_1$  and
$\pi^N(\n\theta,\n\nu)$ the corresponding one for model $f_2$ ($\n\theta$ is of interest if the
reference prior is used). We {\it define} $\pi^N(\n\theta\mid\n\nu)$ such that
$$
\pi^N(\n\theta,\n\nu)=\pi^N(\n\theta\mid\n\nu)\,\pi^N(\n\nu).
$$
To define the DB priors, let $D$ any of \eqref{DSnu} or \eqref{DMnu} (other appropriate divergence
measures could also be explored). Then we define:

\begin{deft}{\bf (DB priors) }\label{piDWn} Let 
$c(q,\n\nu)=\int\, \big(1 + \bar{D}[(\n\theta,\n\theta_0)\mid \n\nu] \big)
^{-q}\,\pi^N(\n\theta\mid\n\nu)\, d\n\theta$, and
$$
\uq =\inf\{\, q\ge 0:\hspace{.1cm} c (q,\n\nu)<\infty\},\hspace{.1cm} \mbox{a.e.}\,
\n\nu\in\Upsilon,\hspace{.6cm} q_* =\uq + 1/2
$$
If $\uq <\infty$, the D-divergence based prior under $M_1$ is $\pi^D_1(\n\nu)=\pi^N(\n\nu),$ and
under $M_2$  is $\pi_2^{D}(\n\theta,\n\nu) = \pi^{D}(\n\theta \mid \n\nu) \ \pi^N(\n\nu)$, where the
(proper) $\pi^{D}(\n\theta \mid \n\nu)$ is
$$
\pi^{D}(\n\theta \mid \n\nu)=c(q_*,\n\nu)^{-1}\,\big(1 + \bar{D}[(\n\theta,\n\theta_0)\mid \n\nu]
\big) ^{-q_*}\,\pi^N(\n\theta\mid\n\nu) \ .
$$
\end{deft}
In this defintion we are implicitly using the reccomended
non-increasing function $h_q(t)=(1+t)^{-q}$, but again other
non-increasing functions on $t\in[0,\infty)$ could be explored.

\begin{deft}\label{piDsDmWn} {\rm \bf (Sum and Minimum DB priors)} The sum DB prior $\pi^S$ and the minimum DB
prior $\pi^M$ are the DB priors given in definition \ref{piDWn} with $D$ being respectively $D^S$
(see \eqref{DSnu}) and $D^M$ (see \eqref{DMnu}). When needed, we refer to their corresponding c's and
q's as $c_S, \uq^S, q_*^S$, and $c_M, \uq^M, q_*^M$, respectively.
\end{deft}



We next investigate whether the DB priors are invariant under reparameterizations. Suppose that
$\n\xi=\n\xi(\n\theta)$  and $\n\eta=\n\eta(\n\nu)$ are, respectively one-to-one monotone mappings
$\n\xi:\Theta\rightarrow\Theta_\xi$, $\n\eta:\Upsilon\rightarrow\Upsilon_\eta$. Clearly, the
reparameterization $(\n\xi,\n\eta)$ preserves orthogonality.

The original problem (\ref{compNuiO}) in this parameterization
becomes: \beq{compNuiOR}
M_1^*:f_1^*(\n{y}\mid\n\eta)=f^*(\n{y}\mid\n\xi_0,\n\eta)\hspace{.25cm}\mbox{vs.}
\hspace{.25cm}
M_2^*:f_2^*(\n{y}\mid\n\xi,\n\eta)=f^*(\n{y}\mid\n\xi,\n\eta), \eeq
where $f^*(\n y\mid\n\xi(\n\theta),\n\eta(\n\nu))=f(\n
y\mid\n\theta,\n\nu)$ and $\n\xi_0=\n\xi(\n\theta_0)$. We next show
that if $\pi^N(\n\nu)$ and $\pi^N(\n\theta,\n\nu)$ are invariant
under these reparameterizations, so are the DB priors. (See Datta
and Ghosh, 1995 for a detailed analysis about the invariance of
several non informative priors in the presence of nuisance
parameters.)

\begin{thm}\label{invNui} (Invariance under one-to-one transformations.)
Let $\pi^D_\nu(\n\nu)$ and $\pi^D_\eta(\n\eta)$ be either the sum or the minimum DB priors under
$M_1$ for the original (\ref{compNuiO}), and reparameterized (\ref{compNuiOR}) problems,
respectively, and similar notation for $\pi_{\theta,\nu}^{D}(\n\theta,\n\nu)$ and
$\pi_{\xi,\eta}^{D}(\n\xi,\n\eta)$, under $M_2$. If $\pi_\nu^N(\n\nu)=\kappa\,
\pi_\eta^N(\n\eta(\n\nu))
 \, |{\cal J}_\eta(\n\nu)|$, where $\kappa$ is a constant, and $\pi_{\theta,\nu}^N(\n\theta,\n\nu)\propto
\pi_{\xi,\eta}^N(\n\xi(\n\theta),\n\eta(\n\nu)) \, |{\cal J}_{\xi,\eta}(\n\theta,\n\nu)|$, then
$$
\pi^D_\nu(\n\nu)=\kappa\, \pi_\eta^D(\n\eta(\n\nu))\, |{\cal J}_\eta(\n\nu)|,\hspace{.5cm}
\pi_{\theta,\nu}^{D}(\n\theta,\n\nu)=\kappa\, \pi_{\xi,\eta}^{D}(\n\xi(\n\theta),\n\eta(\n\nu)) \,
|{\cal J}_{\xi,\eta}(\n\theta,\n\nu)|.
$$
\end{thm}
\begin{proof} See Appendix.
\end{proof}

As a consequence,  DB Bayes factors are not affected by
reparameterizations of the type considered. These are the most
natural and interesting reparameterizations of the problem (and
indeed other reparameterizations seem questionable).
Also, the DB priors are compatible with reduction by sufficiency in
the same spirit as in Proposition~\ref{CSS}.



\subsection{Examples}
We next demonstrate the behavior of  DB priors and corresponding Bayes factors in a couple of
examples. The first is testing the mean of a gamma model, a difficult problem in general. The second
discusses linear models.

\subsubsection{Gamma model (Example 6)}\label{gamSubSub}
Let $\n y=(y_1,\ldots,y_n)$ be an iid sample from a Gamma model with mean $\mu$, and shape parameter
$\alpha$, that is, from
$$
f(y\mid\alpha,\mu)=\big(\frac{\alpha}{\mu}\big)^\alpha \Gamma(\alpha)^{-1}\, y^{\alpha-1}\,
e^{-y\alpha/\mu}.
$$
It is desired to test $H_1:\mu=\mu_0$ vs. $H_2:\mu\ne\mu_0$. It is easy to show that $\mu$ is
orthogonal to $\alpha$.

The objective (reference) priors are $\pi^N(\alpha)=(\psi^{(1)}(\alpha)-1/\alpha)^{1/2}$ and
$\pi^N(\mu,\alpha)=\mu^{-1}(\psi^{(1)}(\alpha)-1/\alpha)^{1/2}$, where $\psi^{(1)}$ represents the
digamma function.  Hence $\pi_2^N(\mu\mid\alpha)=\mu^{-1}$.

The DB priors are $\pi^D(\alpha)=\pi^N(\alpha)$, under both hypotheses and for $D$ either the sum
or min divergence. Under $H_2$, the conditional sum-DB prior for $\mu$ is
$$
\pi^{S}(\mu \mid \alpha)=c_s^{-1}(\alpha)\big[1+\alpha \
\frac{(\mu-\mu_0)^2}{\mu\mu_0}\big]^{-1/2}\, \frac{1}{\mu}
$$
where $c_s(\alpha)$ is the proportionality constant
$$
c_s(\alpha)=\int_0^\infty\, \big(1+\alpha\frac{(t-1)^2}{t}\big)^{-1/2}\,\frac{1}{t}\, dt.
$$
The conditional min-DB prior is
 \ba \pi^{M}(\mu \mid \alpha)=c_m(\alpha)^{-1}\,
  \big(1+\bar{D}^M[(\mu,\mu_0)\mid\alpha]\big)^{-3/2}  \frac{1}{\mu} 
 \ea
where \ba \bar{D}^M[(\mu,\mu_0)\mid\alpha]&=&\Big\{
\begin{array}{ccc}
2\,\alpha(\log\frac{\mu}{\mu_0}-1+\frac{\mu_0}{\mu}) & \mbox{if} & \mu>\mu_0\\
2\,\alpha(\log\frac{\mu_0}{\mu}-1+\frac{\mu}{\mu_0}) & \mbox{if} & \mu\le\mu_0,
\end{array}
\ea
and
$$
c_m(\alpha)=2\,\int_0^\infty\, \big(1+2\alpha(t-1+e^{-t})\big)^{-3/2}\,dt.
$$

In Table~\ref{GamTable} we show the corresponding Bayes factors $B^S_{12}$ and $B^M_{12}$ for
$n=10$; the null value is $\mu_0=10$, and we have considered several combinations of
$(\hat\mu,\hat\sigma)$, the maximum likelihood estimates of the mean and standard deviation. When
$\hat\mu=12$ (casting doubt on the null), both Bayes factors are very similar, and increasing with
$\hat\sigma$, an intuitive behavior. When the data shows the most support for the null, that is,
when $\hat\mu=10$, the Bayes factors differ, with the sum-DB prior giving the most support to the
null.

\begin{table}[t!]
\begin{center}{\footnotesize
\begin{tabular}{c|cc|cc|cc}
& \multicolumn{2}{c}{$\hat\mu=10$} & \multicolumn{2}{c}{$\hat\mu=11$} & \multicolumn{2}{c}{$\hat\mu=12$}\\
& $B_{12}^S$ & $B_{12}^M$ & $B_{12}^S$ & $B_{12}^M$ &  $B_{12}^S$ & $B_{12}^M$ \\
\hline
$\hat\sigma=0.5$ & 12.94 & 2.83 & 0.005 & 0.004 & 1$\cdot 10^{-5}$ & 3$\cdot 10^{-5}$\\
$\hat\sigma=1$   & 11.27 & 2.92 & 0.353 & 0.150 & 0.003 & 0.003\\
$\hat\sigma=2$   & 9.49 & 3.06 & 3.102 & 1.136 & 0.22 & 0.12\\
\hline
\end{tabular}\caption{Values of $B_{12}$ for gamma mean testing with $\mu_0=10$; we use $n=10$, and different values of $(\hat\mu,\hat\sigma)$ .}\label{GamTable}}
\end{center}
\end{table}

In contrast with DB priors, it is not possible to derive relatively
simple expressions for the intrinsic priors. Hence, in this example,
we compare the DB Bayes factors with the intrinsic arithmetic Bayes
factor $IB_{12}^A$ (see Berger and Pericchi 1996). Although
$IB_{12}^A$ does not exactly correspond to a Bayes factor derived
from a specific prior, it does asymptotically correspond to a Bayes
factor derived with the intrinsic arithmetic prior. Since
$IB_{12}^A$ is not defined with reduction by sufficiency, the
comparison are carried out for (specific) simulated samples with the
given parameters. In Table~\ref{Gam2Table} we show the arithmetic
intrinsic and DB Bayes factors for testing $H_1:\mu=10$, with $n=10$
and samples generated from Gamma distributions with
$\mu\in\{10,11,12\}$ and  $\sigma\in\{0.5, 1.0, 2.0\}$. The
resulting MLE´s $(\hat\mu, \hat\sigma)$ in lexicographical order
are: \{(10.02,0.52), (9.98,0.99), (9.98,1.97), (11.01,0.48),
(11.00,0.99), (10.98,1.99), (11.99,0.51), (11.98,0.99),
(12.01,1.99)\}. When $H_2$ is true ($\mu=11$ or $\mu=12$), the three
measures are rather close. Similar values are also obtained when the
`null' model $H_1$ is true and $\sigma=2$. In all these cases, the
three measures provide support to the true model. Nevertheless, when
$H_1$ is true and the variance is small, the DB Bayes factors are
very sensible (with $B_{12}^S$ giving the largest support to the
null) but the $IB_{12}^A$ is not, giving support to $H_2$. This
behavior of $IB_{12}^A$ is likely due to the well known instability
of $IB_{12}^A$ when the sample size is small (worsened in this case
because the variance is small).

\begin{table}[t!]
\begin{center}{\footnotesize
\begin{tabular}{c|ccc|ccc|ccc|}
& \multicolumn{3}{c}{$\mu=10$} & \multicolumn{3}{c}{$\mu=11$} & \multicolumn{3}{c}{$\mu=12$}\\
& $B_{12}^S$ & $B_{12}^M$ & $IB_{12}^A$ & $B_{12}^S$ & $B_{12}^M$ & $IB_{12}^A$ & $B_{12}^S$ & $B_{12}^M$ & $IB_{12}^A$\\
\hline
$\sigma=0.5$ & 13.17 & 2.93 & 0.08 & 0.004 & 0.003 & 0.001 & 1.4$\cdot 10^{-5}$ & 3.7$\cdot 10^{-5}$ & 0.1$\cdot 10^{-5}$\\
$\sigma=1$  & 11.15 & 2.88 & 0.55 & 0.33 & 0.14 & 0.07 & 0.003 & 0.003 & 0.001 \\
$\sigma=2$  & 9.57 & 3.08 & 3.71 & 3.07 & 1.12 & 1.23 & 0.22 & 0.12 & 0.07\\
\hline
\end{tabular}\caption{For Gamma model problem, and test $H_1:\mu=10$ vs. $H_2:\mu\ne 10$. In each cell,
values of $B_{12}$ and arithmetic intrinsic Bayes factor $IB_{12}^A$, associated with a sample of
size $n=10$, from a Gamma model with mean $\mu$ and standard deviation
$\sigma$.}\label{Gam2Table}}
\end{center}
\end{table}

\subsubsection{Variable selection in linear models (Example 7).}\label{subsubLM}
We briefly show next the motivating example for this paper; specifically we show how the DB prior
reproduces Jeffreys-Zellner-Siow prior for variable selection in linear models. More elaborated
examples of testing in linear models can be found in Bayarri and Garc\'{\i}a-Donato (2007). Derivations
of DB priors for random effects are given in Garc\'{\i}a-Donato and Sun (2007).

Consider the full rank General Linear Model $\{N_n(\n y\mid \n
X_1\n\beta_1+\n X_e\n\beta_e, \sigma^2 \n I_n)\}$ and the problem of
testing $H_1:\n\beta_e=\n 0$. After the usual orthogonal
reparameterization (see e.g. Zellner and Siow 1984) and taking $n^*=
n$ and $\pi^N(\n\beta_1,\n\beta_e,\sigma)=\sigma^{-1}$, the DB
priors are
$$
\pi_1^D(\n\beta_1,\sigma)=\sigma^{-1},\hspace{.5cm} \pi_2^D(\n\beta_1,\n\beta_e,\sigma)=
\sigma^{-1}Ca_{k_e}(\n\beta_e\mid \n 0,n^* \sigma^2(\n V^{t}\n V)^{-1}),
$$
where $k_e$ is the dimension of $\n\beta_e$ and
$$
\n V=(\n{I}_n-\n{P}_1)\n{X}_e, \hspace{.5cm} \n{P}_1=\n{X}_1(\n{X}_1^{t}\n{X}_1)^{-1}\n{X}_1^{t}.
$$

Note that the exact matching of JZS and DB priors only occur if the effective sample size is $n^*
= n$. This `coincidence' was the original motivation for the specific choice $\uq+1/2$ in the
definition of DB priors (see Garc\'{\i}a-Donato, 2003 for details).
However, $n^*$ might well depend on the design matrix (or covariates). For example, in the linear
model $\n Y = \n X \theta + \n\epsilon$, with $\n X: n \times 1$ and $\theta$ scalar, it is
intuitively clear that if $\n X = (1,\ldots,1)^t$ then $n^*$ should be $n$, but if $\n X =
(1,\varepsilon, \ldots,\varepsilon)^t$ with $\varepsilon$ very small, then $n^*$ should be $1$.
The effective sample size defined in Berger et al. (2007) satisfies this requirement but other
definitions might not. Extended investigation of this issue is beyond the scope of this paper and
will be pursued elsewhere.

Since comparison among existing objective Bayesian testing procedures for the Linear model have
extensively been given in the literature, including Bayes factors derived with JZS priors, we skip
them here (see for example Berger, Ghosh and Mukhopadhyay, 2003; Liang et al., 2007; Bayarri and
Garc\'{\i}a-Donato, 2007).

\section{Approximations and computation}
In this Section, we derive simple approximations to DB priors and show their connections with already
existing proposal. We also exploit the connection between DB Bayes factors and a corrected Bayes factor
computed with usual (possibly improper) non-informative priors to propose easy MCMC computation of DB Bayes factors.
\subsection{Approximated DB priors}
It is well known (see Kullback 1968; Schervish 1995) that the Kullback-Leibler divergence measures can be
approximated up to second order using the expected Fisher information, so that:
$$
D^S[(\n\theta, \n\theta_0)\mid\n\nu]\approx (\n\theta-\n\theta_0)^t\, J_\theta (\n\theta_0,\n\nu)\,
(\n\theta-\n\theta_0)\approx D^M[(\n\theta,\n\theta_0)\mid\n\nu],
$$
where $J_\theta(\n\theta_0,\n\nu)$ is the block in Fisher
information matrix corresponding to $\n\theta$, evaluated at
$(\n\theta_0,\n\nu)$. Hence, for the problem  (\ref{compNuiO})
(recall that $\n\theta$ and $\n\nu$ are orthogonal), the DB priors
$\pi^D$ (either $\pi^{S}$ or $\pi^{M}$) can be approximated by
$\pi^D_1(\n\nu)=\pi^N(\n\nu)$ and

\begin{equation}\label{apDB}
\pi^D(\n\theta \mid \n\nu)=c(q_*,\n\nu)^{-1}\,h_{q_*} \Big((\n\theta-\n\theta_0)^t\,
\frac{J_\theta(\n\theta_0,\n\nu)}{n^*}\, (\n\theta-\n\theta_0)\Big)\, \pi^N(\n\theta \mid \n\nu),
\end{equation}
where now $q_* = \uq + 1/2$, and $\uq$ is the infimum of $q$ values for which the conditional defined
in \eqref{apDB} (in terms of Fisher information) is proper.

 The cases when $\pi^N(\n\theta\mid\n\nu)$ does not depend
on $\n\theta$ (so $\n\theta$ behaves asymptotically  as a location parameter) 
 are specially interesting. It is easy to then show that  $\uq=k/2$, where $k$ is the dimension of $\n\theta$ and hence
\begin{equation}\label{apDBloc}
\pi^D(\n\theta \mid \n\nu) \approx Ca_k(\n\theta\mid\n\theta_0,
n^*\,J^{-1}_\theta(\n\theta_0,\n\nu)) \, , 
\end{equation}
The conditional prior \eqref{apDBloc} has been interpreted by many authors (see for instance Kass and
Wasserman 1995) as the generalization of Jeffreys' ideas to multivariate problems.

Moreover, if $h_q(t)=e^{-qt}$ 
is used instead, then $\pi^D$ would essentially be the normal unit
information priors, as defined by Kass and Wasserman (1995) and
further studied by Raftery (1998). Note that we have shown  that
this proposals can be interpreted as approximated DB priors only
when $\n\theta$ is  asymptotically a location parameter.

\subsection{Computation of Bayes factor}
Interestingly enough, and similarly to other objective Bayesian proposals (like the intrinsic  and
fractional Bayes factors),  it can be shown that Bayes factors computed with DB priors, $B_{21}^D$,
can be expressed as an (invalid) Bayes factor computed with non-informative (usually improper)
priors, $B_{21}^N$,  multiplied by a correction factor. This expression also allows for easy
computation of DB Bayes factors when $B_{21}^N$ is easy to compute.

\begin{lem}\label{MCBF}
For problem (\ref{compNuiO}) (with $\n\theta$ and $\n\nu$ orthogonal), let $B_{12}^N$ denote the
Bayes factor computed using  $\pi_1^N(\n\nu)$ and $\pi_2^N(\n\theta, \n\nu)$, then for both the
sum and min DB-priors
\begin{equation}\label{BE}
B_{21}^{D}= B_{21}^N\, \times E^{\pi^N(\theta,\nu \mid y)}\,
\Big(c(q_*,\n\nu)^{-1}\,h_{q_*}(\bar{D}[(\n\theta,\n\theta_0)\mid\n\nu])\Big).
\end{equation}
\end{lem}
\begin{proof} See Appendix. \end{proof}

Computation of $B_{21}^N$ is often  simpler than computation of proper Bayes factors. Then a
sample (usually MCMC) from the  posterior distribution $\pi^N(\n\theta,\n\nu\mid\n y)$ can be used
to evaluate the expectation in \eqref{BE}, thus considerably simplifying computation of $B_{12}^S$
or $B_{12}^M$. This is actually how we computed the Bayes factors for Example 6 in
Section~\ref{gamSubSub}.

Moreover, if $n$ is large (relative to the dimension of $\n\phi = (\n\theta, \n\nu)$, assumed fixed)
we can approximate \eqref{BE} using asymptotic expressions to posterior distribution along with the
approximated DB priors given in \eqref{apDB}. 

We illustrate the approach in a simple setting. First we assume that the asymptotic posterior
distribution is given by (see conditions in e.g. Berger 1985),
$$
\pi^N(\n\theta,\n\nu\mid\n y)\approx N(\hat{\n\phi},\n J^{-1}(\hat{\n\phi})),
$$
where $\hat{\n\phi}=(\hat{\n\theta},\hat{\n\nu})$ is the  (assumed to exist) maximum likelihood
estimate of $(\n\theta,\n\nu)$ and $\n J=\n J_\theta\oplus \n J_\nu$ is the (block diagonal) expected
Fisher information matrix of $f(\n y\mid\n\theta,\n\nu)$.

Next we assume that $\pi^N(\n\theta\mid\n\nu)$ does not depend on
$\n\theta$, so the approximating (conditional) DB prior is the
Cauchy prior in \eqref{apDBloc}. As a notational device, it will be
convenient to then write $\pi^N(\n\theta\mid\n\nu)$ as
$\pi^N(\n\theta_0\mid\n\nu)$. Expressing the Cauchy density
\eqref{apDBloc} in the usual way as a scale mixture of a Normal and
an inverse gamma, and using the asymptotic posterior, the DB Bayes
factors, as given in \eqref{BE}, can be approximated by
$$
B_{21}^D\approx B_{21}^N\, \int\int
\frac{1}{\pi^N(\n\theta_0\mid\n\phi)}N_k(\hat{\n\theta}\mid\n\theta_0,\Sigma(\n\nu,t))\,
 N_p(\n\nu\mid\hat{\n\nu}, \n J_\nu(\hat{\n\phi}))\,d\n\nu \,
IGa(t\mid\frac{1}{2},\frac{1}{2})\, dt,
$$
where $p$ is the dimension of $\n\nu$ and $\Sigma(\n\nu,t)= t \, n \, \n
J_\theta^{-1}(\n\theta_0,\n\nu)+\n J_\theta^{-1}(\hat{\n\phi})$. A similar asymptotic approximation
to $B_{12}^N$, finally gives the desired asymptotic approximation to the DB Bayes factor:
\begin{eqnarray*}
B^D_{21} &\approx&\frac{p(\n y\mid\hat{\n\phi})}{p(\n
y\mid\n\theta_0,\hat{\n\nu})}(2\pi)^{k/2}\frac{1}{det \n J_\theta(\hat{\n\phi})^{1/2}}\,\\
&\times& \int\int \frac{\pi^N(\hat{\n\theta}\mid\hat{\n\nu})}{\pi^N(\n\theta_0\mid\n\nu)}
N_k(\hat{\n\theta}\mid\n\theta_0,\Sigma(\n\nu,t))\, N_p(\n\nu\mid\hat{\n\nu}, \n J_\nu(\hat{\n\phi}))
IGa(t\mid\frac{1}{2},\frac{1}{2})\,d\n\nu\, dt,
\end{eqnarray*}
which is very easy to evaluate by simple Monte Carlo. Note that arbitrary constants in the possibly
improper $\pi^N(\n\theta \mid \n\nu)$ cancel out in the expression
above. 

\section{Summary and conclusions}
Extending pioneering work by Jeffreys (1961), we propose a new class of priors for objective Bayes
hypothesis testing based on divergence measures, which we call `Divergence Based' (DB) priors. For
divergence measures, we propose use of symmetrized versions (sum and the minimum) of Kullback
Liebler divergences. The resulting DB priors are usually easy to compute and have a number of
desirable properties as invariance under reparameterizations, evidence consistency and
compatibility with sufficient statistics. We explore DB priors in a series of estudy examples, in
which they show to be intuitively sound and to  produce sensible Bayes factors. This is so even
for irregular models and improper likelihoods, which are extremely challenging scenarios for other
objective Bayes testing methodologies. We recommend use of the sum-DB prior when it exists because
it is considerably easier to compute than the min-DB prior and seems to exhibit a nicer behavior.

The DB priors seem to behave similarly to the arithmetic intrinsic prior (when defined). Also, in
normal scenarios, they exactly reproduce the proposals of Jeffreys (1961) and Zellner and Siow
(1980, 1984), so that they can be considered an extension of these classical proposals to
non-normal situations. Approximations to DB priors are also shown to be connected with other
proposals as the unit information priors. Finally, we also provide asymptotic approximations to DB
Bayes factors for large sample size.

The definition of DB priors are based on particular choices of both
{\it 1)} an `objective prior' $\pi^N$ for estimation problems and
{\it 2)} an equivalent sample size $n^*$. Of course, there is no
general agreement in the literature about a single definition for
any of these concepts (and there might never be). We think that any
sensible proposals would produce nice results, but this in an issue
that needs to be further investigated. We recommend, when possible,
use of the {\it reference prior} (Berger and Bernardo, 1992) and of
the equivalent sample size in Berger et al. (2007).

Other apparently arbitrary choices that we made were those of $h_q$ and of $q_*$, however they were
based on some compelling arguments

\begin{itemize}
\item Choice of $h_q(t)=(1+t)^{-q}$ was specifically chosen to reproduce in the normal case
Jeffreys-Zellner-Siow priors, but there are other reasons for it. A compelling reason is that it
is a simple function resulting in Bayes factors with nice properties; another simple function to
use could be the exponential, but this results in normal priors that are not {\it evidence
consistent}. Also, $h_q$ results in priors with very heavy tails, which is important so as not to
`knock-out' the likelihood when data is not well explained by the null model. However, we do not
rule out that other choices  of  functions  $h(t)$  which are decreasing for $t\in[0,\infty)$,
with maximum at zero, and producing proper DB-type priors could work better in specific scenarios.

\item Choice of $q^* = \uq +1/2$. In principle, any $\uq +\delta$ could be used. As a matter of fact,
we do not expect that the specific choice of $\delta$ matters much as long as  $\delta\in(0,1)$
(needed to produce priors with heavy tails and no moments), but this again needs further
investigation. We recommend use of  $\delta=1/2$ because it is the value reproducing Jeffreys
proposal. 
\end{itemize}

\section*{Acknowledgements}
Comments by Jim Berger are gratefully acknowledged. This research was supported in part by the Spanish Ministry of Science and Technology, under Grant MTM2004-03290.

\section*{References}

\vspace{.5cm}\noindent Bayarri, M.J. and Garc\'{\i}a-Donato, G. (2007), ``Extending Conventional priors
for Testing General Hypotheses in Linear Models,'' {\it Biometrika,} {\bf 94,} 135-152.

\vspace{.5cm}\noindent Berger, J.O. (1985),
  {\it Statistical Decision Theory and Bayesian Analysis (2nd ed.),} New York:
  Springer-Verlag.

\vspace{.5cm}\noindent Berger, J. O. and Bernardo, J. M. (1992), ``On the development of the
reference prior method.''. In {\it Bayesian Statistics 4} (eds J. M. Bernardo, J. O. Berger, A. P.
Dawid and A. F. M. Smith), pp. 35-60. Oxford: Oxford University Press.

\vspace{.5cm}\noindent Berger, J.O. and Delampady, M. (1987), ``Testing precise hypotheses,'' {\it
Statistical Science}, {\bf 3,} 317-352.

\vspace{.5cm}\noindent Berger, J.O. and Mortera, J. (1999), ``Default Bayes Factors for Nonnested
Hypothesis Testing,'' {\it Journal of the American Statistical Association}, {\bf 94}, 542-554.

\vspace{.5cm}\noindent Berger, J.O., Ghosh, J.K. and Mukhopadhyay, N. (2003), ``Approximations to
the Bayes factor in model selection problems and consistency issues,'' {\it Journal of Statistical
Planning and Inference,} {\bf 112}, 241-58.

\vspace{.5cm}\noindent Berger, J. O. and Pericchi, L. R. (1996), ``The intrinsic Bayes factor for
model selection and prediction,'' {\it Journal of the American Statistical Association,} {\bf 91},
109-22.

\vspace{.5cm}\noindent Berger, J. O., and Pericchi, R. L. (2001), ``Objective Bayesian methods for
model selection: introduction and comparison (with discussion)''. In {\it Model Selection} (ed P.
Lahiri), pp. 135-207. Institute of Mathematical Statistics Lecture Notes-Monograph Series, volume
38. Beachwood Ohio.

\vspace{.5cm}\noindent Berger, J. O., Pericchi, L. R. and Varshavsky, J. A. (1998), ``Bayes
factors and marginal distributions in invariant situations,'' {\it Sankhya A,} {\bf 60}, 307-21.

\vspace{.5cm}\noindent Berger, J.O. and Sellke, T. (1987), ``Testing a point null hypothesis: the
irreconcilability of P-values
  and evidence,''
  {\it Journal of the American Statistical Association,}
  {\bf 82}, 112-122.

\vspace{.5cm}\noindent Berger, J. et al. (2007). ``Extensions and generalizations of BIC'', ISDS
Working paper, in preparation.

\vspace{.5cm}\noindent Bernardo, J.M. and Rueda, R. (2002), ``Bayesian hypothesis testing: A
reference approach,'' {\it International Statistical Review,} {\bf 70}, 351-372.

\vspace{.5cm}\noindent Bernardo, J.M. (2005), ``Intrinsic credible regions: An objective Bayesian
approach to interval estimation,'' {\it Test}, {\bf 14}, 317-384.

\vspace{.5cm}\noindent Clyde, M. (1999), ``Bayesian Model Averaging and Model Search Strategies
(with discussion)''. In {\it Bayesian Statistics 6} (eds J.M. Bernardo, A.P. Dawid, J.O. Berger,
and A.F.M. Smith), pp. 157-185. Oxford: Oxford University Press.

\vspace{.5cm}\noindent Clyde, M., DeSimone, H. and Parmigiani, G. (1996), ``Prediction via
Orthogonalized Model Mixing,'' {\it Journal of the American Statistical Association}, {\bf 91},
1197-1208.

\vspace{.5cm}\noindent Conover, W. J. (1971), {\it Practical nonparametric statistics,} New York:
John Wiley and Sons.

\vspace{.5cm}\noindent Cox, D. R. and Reid, N. (1987), ``Parameter orthogonality and approximate
conditional inference,'' {\it Journal of the Royal Statistical Society B}, {\bf 49}, 1-39.

\vspace{.5cm}\noindent Datta, G.S. and Ghosh, M. (1995), ``On the invariance of noninformative
priors'', {\it Annals of Statistics}, {\bf 24}, 141-159.

\vspace{.5cm}\noindent De Santis, F. and Spezzaferri, F. (1999), ``Methods for Default and Robust
Bayesian Model Comparison: The Fractional Bayes Factor Approach,'' {\it International Statistics
Review,} {\bf 67}, 267-286.

\vspace{.5cm}\noindent Garc\'{\i}a-Donato, G. (2003),
   {\it Factores Bayes Factores Bayes Convencionales: Algunos
    Aspectos Relevantes,} Unpublished PhD Thesis,
   Department of Statistics, University of Valencia.

\vspace{.5cm}\noindent Garc\'{\i}a-Donato, G. and Sun, D. (2007), ``Objective Priors for Model
Selection in One-Way Random Effects Models,'' {\it The Canadian journal of Statistics,} in press.

\vspace{.5cm}\noindent Hoeting, J.A, Madigan, D., Raftery, A.E. and Volinsky, C.T. (1999),
``Bayesian Model Averaging: A Tutorial,'' {\it Statistical Science,} {\bf 14}, 382-417.

\vspace{.5cm}\noindent Ibrahim, J. and Laud, P. (1994), ``A Predictive Approach to the Analysis of
Designed Experiments,'' {\it Journla od the American Statistical Association,} {\bf 89}, 309-319

\vspace{.5cm}\noindent Jeffreys, H. (1961). {\it Theory of Probability}, 3rd edn. London: Oxford
University Press.

\vspace{.5cm}\noindent Kass, R. E. and Raftery, A. E. (1995), ``Bayes factors,'' {\it Journal of
the American Statistical Association,} {\bf 90}, 773-95.

\vspace{.5cm}\noindent Kass, R. E. and Vaidyanathan, S. (1992), ``Approximate Bayes factors and
orthogonal parameters, with application to testing equality of two binomial proportions,'' {\it
Journal of the Royal Statistical Society B,} {\bf 54}, 129-44.

\vspace{.5cm}\noindent Kass, R. E. and Wasserman, L. (1995), ``A reference Bayesian test for
nested hypotheses and its relationship to the Schwarz criterion,'' {\it Journal of the American
Statistical Association,} {\bf 90}, 928-34.

\vspace{.5cm}\noindent Kullback, S. (1968), {\it Information Theory and Statistics,} New York:
Dover Publications, Inc.

\vspace{.5cm}\noindent Laud, P.W. and Ibrahim, J. (1995), ``Predictive Model Selection,'' {\it
Journal of the Royal Statistical Society B,}{\ bf 57}, 247-262.

\vspace{.5cm}\noindent Liang, F., Paulo, R., Molina, G., Clyde, M., and
  Berger, J. O. (2007),
   ``Mixtures of g -priors for Bayesian Variable Selection,''
   {\it Journal of the American Statistical Society}, in press.

\vspace{.5cm}\noindent Montgomery, D. (2001), {\it Introduction to Statistical Quality Control},
4th edn. John Wiley and Sons, Inc.

\vspace{.5cm}\noindent Moreno, E., Bertolino, F. and Racugno, W. (1998), ``An intrinsic limiting
procedure for model selection and hypotheses testing,'' {\it Journal of the American Statistical
Association,} {\bf 93}, 1451-60.

\vspace{.5cm}\noindent O'Hagan, A. (1995),
 ``Fractional Bayes factors for model comparison (with discussion),''
 {\it Journal of the Royal Statistical Society, B},
 {\bf 57,} 99-138.

\vspace{.5cm}\noindent Pauler, D. (1998), ``The Schwarz Criterion and Related Methods for Normal
Linear Models,'' {\it Biometrika,} {\bf 85}, 13-27.

\vspace{.5cm}\noindent Pauler, D.K., Wakefield, J.C. and Kass, R.E. (1999), ``Bayes factors and
approximations for variance component models,'' {\it Journal of the American Statistical
Association},
 {\bf 94,} 1242-1253.

\vspace{.5cm}\noindent P\'{e}rez, J.M. and Berger, J. (2001), ``Analysis of mixture models using
expected posterior priors, with application to classification of gamma ray bursts.'' In {\it
Bayesian Methods, with applications to science, policy and official statistics,} (eds E. George
and P. Nanopoulos), pp. 401-410. Official Publications of the European Communities, Luxembourg.

\vspace{.5cm}\noindent P\'{e}rez, J. M. and Berger, J. O. (2002), ``Expected posterior prior
distributions for model selection,'' {\it Biometrika,} {\bf 89}, 491-512.

\vspace{.5cm}\noindent P\'{e}rez, S. (2005), {\it M\'{e}todos Bayesianos objetivos de comparaci\'{o}n de
medias,} Unpublished PhD Thesis,    Department of Statistics, University of Valencia.

\vspace{.5cm}\noindent Raftery, A.E. (1998), ``Bayes factor and BIC: comment on Weakliem,''
Technical Report 347, Department of Statistics, University of Washington.


\vspace{.5cm}\noindent Schervisch, M.J. (1995), {\it Theory of Statistics.} New York:
Springer-Verlag.

\vspace{.5cm}\noindent Tanner, M.A. (1996), {\it Tools for Statistical Inference. Methods for the
exploration of Posterior Distributions and Likelihood Functions.} 3rd edn. New York: Springer
Verlag.

\vspace{.5cm}\noindent Zellner, A. and Siow, A. (1980), ``Posterior odds ratio for selected
regression hypotheses''. In {\it Bayesian Statistics 1} (eds J. M. Bernardo, M. H. DeGroot, D. V.
Lindley and A. F. M. Smith), pp. 585-603. Valencia: University Press.

\vspace{.5cm}\noindent Zellner, A. and Siow, A. (1984). {\it Basic Issues in Econometrics.}
Chicago: University of Chicago Press.


\section*{Appendix. Proofs.}


\noindent\underline{Proof of Proposition~\ref{inv}.} Let
$\bar{D}^*[\n\xi,\n\xi_0]$ be the unitary measure of divergence
between $f_1^*(\n y)$ and $f_2^*(\n y\mid\n\xi)$ in (\ref{compWNr}).
It is well known that $KL$ remains the same under one-to-one
reparameterizations, and clearly
$\bar{D}^*[\n\xi(\n\theta),\n\xi(\n\theta_0)]=\bar{D}[\n\theta,\n\theta_0]$.
Now, by definition of DB priors, and using the relation between
$\pi^N_\theta$ and $\pi^N_\xi$, it follows that
\begin{eqnarray*}
\pi^D_\theta(\n\theta)&=&c(q_*)\, h_{q_*}(\bar{D}[\n\theta,\n\theta_0])\, \pi^N_\theta(\n\theta)=
c(q_*)\,h_{q_*}(\bar{D}^*[\n\xi(\n\theta),\n\xi(\n\theta_0)]) \pi^N_\xi(\n\xi(\n\theta)) |{\cal
J}_\xi(\n\theta)|\\
&=&\pi^D_\xi(\n\xi(\n\theta))\,|{\cal J}_\xi(\n\theta)|.
\end{eqnarray*}

\bigskip

\noindent\underline{Proof of Proposition~\ref{CSS}.} Let
$D^*[\n\theta,\n\theta_0]$ be the symmetric divergence between
$f_1^*(\n t)$ and $f_2^*(\n t\mid\n\theta)$ in (\ref{suf}), and
hence $D^*[\n\theta,\n\theta_0]=D[\n\theta,\n\theta_0]$. The result
now follows from the assumption that neither $\pi^N$
nor $n^*$ change when the problem is formulated in terms of sufficient statistics.\\[.4cm]

\noindent\underline{Proof of Lemma~\ref{expCon}.} First we show that
\eqref{condExp} implies that $B_{12}^\pi\rightarrow 0$ as
$\bar{y}\rightarrow 0$. Assume $\int_0^1\,\mu^{-k}\pi(\mu)=\infty$.
Then
$$
\lim_{\bar{y}\rightarrow 0}\, m_2(\n y)= \lim_{\bar{y}\rightarrow 0} \int_0^\infty \mu^{-n}
e^{-n\bar{y}/\mu}\,\pi(\mu)d\mu\ge \int_0^1\, \mu^{-k}\pi(\mu)d\mu=\infty,
$$
and the result follows. To show the converse, note that, since
$\pi(\mu)$ is proper,
\begin{equation}\label{proper}
\lim_{\bar{y}\rightarrow 0} \int_1^\infty \mu^{-n} e^{-n\bar{y}/\mu}\,\pi(\mu)d\mu<\infty.
\end{equation}
Now, by contradiction suppose that for $n\ge k$, $\int_0^1\,
\mu^{-k}\pi(\mu)d\mu<\infty$, so in particular $\int_0^1\,
\mu^{-n}\pi(\mu)d\mu<\infty$, and hence the limit ing function
$g(\mu)=\mu^{-n}\pi(\mu)$ is integrable; now, the Dominated
Convergence Theorem gives
$$
\lim_{\bar{y}\rightarrow 0} \int_0^1 \mu^{-n} e^{-n\bar{y}/\mu}\,\pi(\mu)= \int_0^1
\mu^{-n}\pi(\mu)d\mu<\infty,
$$
which jointly with (\ref{proper}) contradicts the assumption of
$B_{12}^\pi\rightarrow 0$ as $\bar{y}\rightarrow 0$, proving the result.\\[.4cm]

\noindent\underline{Proof of Lemma~\ref{ecIM}.} It can easily be
seen that, as $T\rightarrow\infty$
$$
B_{21}^\pi\rightarrow e^{-n\theta_0}\,\int_{-\infty}^\infty\, e^{n\theta}\pi(\theta)d\theta,
$$
Now, $\forall n\ge k$, it follows that
$$
\int_{-\infty}^\infty\, e^{n\theta}\pi(\theta)d\theta\ge \int_{-\infty}^\infty\,
e^{k\theta}\pi(\theta)d\theta\ge \int_{\theta_0}^\infty\, e^{k\theta}\pi(\theta)d\theta,
$$
proving the lemma.\\[.5cm]

\noindent\underline{Proof of Theorem~\ref{invNui}.} By definition, the DB priors for the
reparameterized problem are $\pi^D_\nu(\n\nu)=\pi^N_\nu(\n\nu)$ and (recall $h_q(t)=(1+t)^{-q}$)
$$
\pi^D_{\xi,\eta}(\n\xi,\n\eta)=c^*(q_*,\n\eta)^{-1}h_{q^*}(\bar{D}^*[(\n\xi,\n\xi_0)\mid\n\eta])\,
\pi^N_{\xi\mid\eta}(\n\xi\mid\n\eta)\pi^N_\eta(\n\eta),
$$
where $\bar{D}^*[(\n\xi,\n\xi_0)\mid\n\eta]$ is the corresponding unitary measure of divergence
between the competing models $f_1^*$ and $f_2^*$ in (\ref{compNuiOR}) and
$$
c^*(q_*,\n\eta)=\int\,h_{q^*}(\bar{D}^*[(\n\xi,\n\xi_0)\mid\n\eta])\,
\pi^N_{\xi\mid\eta}(\n\xi\mid\n\eta)d\n\xi.
$$
It can be easily shown that
$\bar{D}^*[(\n\xi,\n\xi_0)\mid\n\eta]=\bar{D}[(\n\theta,\n\theta_0)\mid\n\nu]$.
Also, under the assumptions of the theorem,
$\pi_{\theta,\nu}^N(\n\theta,\n\nu)=\kappa_2\,
\pi_{\xi,\eta}^N(\n\xi(\n\theta),\n\eta(\n\nu)) \, |{\cal
J}_{\xi,\eta}(\n\theta,\n\nu)|$, where $\kappa_2$ is a constant.
Then
$$
\pi^N_{\theta\mid\nu}(\n\theta,\n\nu)=\frac{\kappa_2}{\kappa}\,
\pi^N_{\xi\mid\eta}(\n\xi(\n\theta)\mid\n\eta(\n\nu))\,|{\cal J}_{\xi}(\n\theta)|,
$$
and hence
$$
c(q_*,\n\nu)=\frac{\kappa_2}{\kappa}\,c^*(q_*,\n\eta(\n\nu)),
$$
and the result follows.\\[.5cm]

\noindent\underline{Proof of Lemma~\ref{MCBF}.} For $i=1,2$, let $m_i^D(\n y)$ and $m_i^N(\n y)$
denote the prior predictive marginals obtained with $\pi_i^D$ and $\pi_i^N$, respectively. By
definition of DB priors, $m_i^N(\n y)=m_i^D(\n y)$, and hence
$$
B_{21}^D=\frac{m_2^D(\n y)}{m_1^D(\n y)}=\frac{m_2^N(\n y)}{m_1^N(\n y)}\, \frac{m_2^D(\n
y)}{m_2^N(\n y)}=B_{21}^D\, \frac{m_2^D(\n y)}{m_2^N(\n y)}.
$$
Finally
\begin{eqnarray*}
m_2^D(\n y)&=&\int f(\n y\mid\n\theta,\n\nu)\, \pi^D(\n\theta,\n\nu)d\n\theta d\n\nu\\
&=& \int f(\n y\mid\n\theta,\n\nu)\, c(q_*,\n\nu)^{-1}
h_{q_*}(\bar{D}[(\n\theta,\n\theta_0)\mid\n\nu]) \pi^N(\n\theta,\n\nu) d\n\theta d\n\nu\\
&=& m_2^N(\n y)\,\int c(q_*,\n\nu)^{-1} h_{q_*}(\bar{D}[(\n\theta,\n\theta_0)\mid\n\nu])
\pi^N(\n\theta,\n\nu\mid \n y) d\n\theta d\n\nu,
\end{eqnarray*}
and the result holds.

%

\end{document}